\documentclass{article}

\usepackage[top=2.5cm, bottom=2.5cm, left=2.5cm, right=2.5cm]{geometry}

\usepackage{etoolbox}
\newtoggle{bigFig}
\toggletrue{bigFig}

\usepackage{xfrac}
\usepackage{lmodern}
\usepackage{anyfontsize}
\usepackage{amsmath}
\usepackage{bm}
\usepackage{bbm}
\usepackage{amsthm}
\usepackage{mathtools}
\usepackage{amsfonts}
\usepackage{stmaryrd}
\usepackage{graphicx}
\usepackage{tabu}
\usepackage{diagbox}
\usepackage{pgfplots}
\pgfplotsset{width=7cm,compat=1.9}
\usepackage{tikz}
\usetikzlibrary{arrows}
\usetikzlibrary{shapes}
\usetikzlibrary{arrows,positioning,shapes,fit,calc}
\pgfdeclarelayer{background}
\pgfsetlayers{background,main}

\usepackage{xcolor}
\colorlet{ggreen}{green!50!black}

\usepackage{subdepth}
\newcolumntype{L}{>{$}l<{$}}

\DeclareMathOperator{\supp}{supp}

\DeclareMathOperator{\awaeOp}{\mathsf{awae}}

\newcommand{\awae}[2]{\awaeOp_{{#1}}^{{#2}}}

\DeclareMathOperator{\HH}{H}
\DeclareMathOperator{\V}{V}

\DeclareMathOperator{\Awae}{\mathsf{Awae}}
\DeclareMathOperator{\obj}{obj}
\DeclareMathOperator{\OP}{\mathsf{OP}}

\DeclareMathOperator{\id}{\mathsf{id}}

\newcommand{\ourparagraph}[1]{{\smallskip\noindent {\bf {#1}.}}}

\renewcommand\vec{\mathbf}
\renewcommand\phi{\varphi}
\renewcommand{\epsilon}{\varepsilon}

\newcommand{\XA}{X_\mathbb{A}}

\newcommand{\XT}{X_\mathbb{T}}

\newcommand{\XS}{X_\mathbb{S}}

\newcommand{\xA}{\vec{x_\mathbb{A}}}

\newcommand{\xT}{\vec{x_\mathbb{T}}}
\newcommand{\xS}{\vec{x_\mathbb{S}}}
\newcommand{\xx}{\vec{x}}

\newcommand{\DA}{D_\mathbb{A}}

\newcommand{\DT}{D_\mathbb{T}}
\newcommand{\DS}{D_\mathbb{S}}
\newcommand{\DO}{D_O}
\newcommand{\DOp}{D_{O'}}
\newcommand{\DP}{D_{\Phi}}

\newcommand{\bP}{\mathbb{P}}
\newcommand{\bA}{\mathbb{A}}

\newcommand{\bT}{\mathbb{T}}

\newcommand{\bS}{\mathbb{S}}

\newcommand{\bV}{\mathbb{V}}

\newcommand{\cW}{\mathcal{W}}
\newcommand{\cX}{\mathcal{X}}
\newcommand{\pA}{\pi_{\bA}}
\newcommand{\pT}{\pi_{\bT}}
\newcommand{\pS}{\pi_{\bS}}
\newcommand{\pP}{{\pi_{\Phi}}}

\newcommand{\Z}{\mathbb{Z}}
\newcommand{\R}{\mathbb{R}}
\newcommand{\N}{\mathbb{N}}
\newcommand{\cP}{\mathcal{P}}

\newcommand{\Np}{\mathbb{N}_{> 0}}
\newcommand{\Rpe}{\mathbb{R}_{> 0}}
\newcommand{\Rp}{\mathbb{R}_{\geq 0}}

\newcommand{\plus}{\scalebox{.6}{$+$}}
\newcommand{\tid}[1]{{#1}^{\plus}}
\newcommand{\abs}[1]{\lvert {#1} \rvert}

\newcommand{\llrr}[2]{\llbracket {#1}, {#2} \rrbracket}
\newcommand{\norm}[1]{\left\lVert#1\right\rVert}

\newcommand{\vv}[1]{\left\langle {#1} \right\rangle}
\newcommand{\card}[1]{\left|{#1}\right|}

\newtheorem{example}{Example}
\newtheorem{theorem}{Theorem}
\newtheorem{corollary}{Corollary}
\newtheorem{lemma}{Lemma}
\newtheorem{definition}{Definition}
\newtheorem{assumption}{Assumption}
\newtheorem{method}{Method}

\begin{document}

\begin{center}
{\Large \bf Optimal Accuracy-Privacy Trade-Off for\\ Secure Multi-Party Computations}\\

\bigskip
\bigskip
{\small Patrick Ah-Fat and Michael Huth\\
Department of Computing, Imperial College London\\
London, SW7 2AZ, United Kingdom\\
$\{$patrick.ah-fat14, m.huth$\}$@imperial.ac.uk}
\end{center}

\date{\today}

\bigskip
\begin{abstract}
The purpose of Secure Multi-Party Computation is to enable protocol participants to compute a public function of their private inputs while keeping their inputs secret, without resorting to any trusted third party. 
However, opening the public output of such computations inevitably reveals some information about the private inputs. 
We propose a measure generalising both R\'enyi entropy and $g$-entropy so as to quantify this information leakage. 
In order to control and restrain such information flows, we introduce the notion of function substitution which replaces the computation of a function that reveals sensitive information with that of an approximate function. 
We exhibit theoretical bounds for the privacy gains that this approach provides and experimentally show that this enhances the confidentiality of the inputs while controlling the distortion of computed output values. 
Finally, we investigate the inherent compromise between accuracy of computation and privacy of inputs and we demonstrate how to realise such optimal trade-offs. 
\end{abstract}


\section{Introduction}
We study the setting of functions $f$ that map $n$ integral inputs $x_1,\dots,x_n$ into one integral output. Each input $x_i$ is controlled by some agent $i$ and its value is considered private to agent $i$. The computation of function $f$ is \emph{multi-party secure} if its evaluation protects the privacy of the inputs, so that agent $j$ cannot learn more from this computation about the other values $x_i$ than what agent $j$ is able to infer from knowledge of her own input $x_j$ and the publicly observable output $f(x_1,\dots,x_n)$.

Secure Multi-party Computation (SMC) is a domain of cryptography that can implement such a black-box functionality: it enables protocol participants to compute a public function of their private inputs, such that no trusted third party is required, and that the confidentiality of the inputs is protected \cite{yao1986generate,yao1982protocols,shamir1979share,%
rabin1989verifiable,ben1988completeness,chaum1988multiparty}. 
Recent advances in SMC have given birth to a variety of efficient protocols that achieve computational and information-theoretic security against passive and active adversaries \cite{kolesnikov2008improved,lindell2015efficient,lindell2012secure,%
araki2016high}. 

SMC therefore gives strong security guarantees, but it does allow inferences about other agents' input values based on the publicly observable output and one's own private input. This is referred to as the \emph{acceptable information flow} in the SMC literature, in which this is therefore largely ignored. In fact, this so called acceptable information flow is oblivious to the manner in which a protocol realises the aims of SMC. 

Consequently, such information flows would also occur in the setting of outsourced computation. In this case, a trusted third party or a central authority (e.g. a national health agency) holds some records from different parties (e.g. some medical insurance companies), computes a function of those records and informs the parties of the result of the computation, such that no information leaks about the parties' inputs apart from the public output. 

But we believe that such information flow is not always acceptable, e.g., in the medical domain with its strict privacy regulations. Moreover, we think that it is important to understand and quantify such information flow in order to
\begin{itemize}
\item better understand potential risks of using SMC in a specific application, say, a health-care consortium of insurers and hospitals

\item devise methods that can mitigate or prevent such information flow.
\end{itemize}

\noindent The latter aim contains within it an inherent friction. 
The information flow 
whose existence only depends on knowledge of some private inputs and the public output can neither be mitigated against nor prevented by an SMC protocol that computed that function $f$.
Rather, for such measures to be effective, we will need to modify the actual behaviour of function $f$: we will instead use another function $f'$ for which the acceptable information flow is absent, less pronounced or optimal according to some risk measure. 
The aforementioned friction consists of the need to shield against such undesired information flow for function $f$ by replacing the latter with function $f'$. This substitution naturally introduces some inaccuracy in the value of the computed output, which will need to be controlled. 

The notions of security developed for SMC are not directly helpful in understanding this friction and its inherent trade-off. SMC security neither reflects the amount of information that leaks from a computation once the public output is revealed, nor does it account for the ability of an attacker to influence such leakage before entering a protocol \cite{lindell2009secure,orlandi2011multiparty,%
cramer2015secure,aumann2007security}. We therefore develop, in this paper, bespoke methods for understanding this better.
Specifically:
\begin{enumerate}
\item We generalise a model of such information flow, developed in \cite{ah2017secure}, to an entire family of conditional entropies. This subsumes the Shannon and min entropies as well as the notion of $g$-leakage in Computer Security.

\item We devise a method of distorting  the output of a function $f$ with so called \emph{virtual inputs} so that the distorted function $f'$ may be computed through standard means, such as SMC protocols.

\item We express a trade-off between privacy preservation of agents' inputs and output accuracy as a non-linear optimisation problem, whose solution computes optimal virtual inputs.

\item We demonstrate that these optimisations can be solved for a large class of our conditional entropies, including the aforementioned ones.

\item We also offer theoretical insights that relate and characterise the relationship between the accuracy of the distorted function and the level of privacy that such distortions offer.
\end{enumerate}

This work is motivated by, and applicable to, Secure Multi-Party Computations. Our methods do not rely on the particular protocols used for SMC, but only on the abstract setting of a \emph{black-box} function $f$ to which parties $i$ submit a private input $x_i$ and all then learn the public output of $f$. We will therefore present the core of our technical development in this abstract setting, to stress that these results are orthogonal to the choice of an SMC protocol.

Naturally, the application context of an SMC may constrain or inform our approach. In a voting protocol, e.g., a mere deviation from the original function $f$ would hardly be tolerable. But our approach may be used to enhance the privacy in less restrictive scenarios such as in the computation of statistical measures or financial audits.

\ourparagraph{Outline of Paper} We discuss related work in Section~\ref{section:related}. Needed technical background from information theory is covered in Section~\ref{section:background}. Our development of a generalised conditional entropy is the subject of Section~\ref{section:entropy}. The development of our model for information flow for black-box functions and the attacker's entropy for that are described in Section~\ref{section:model}. The method by which one can randomise black-box functions through virtual inputs is developed in Section~\ref{section:random} and its theory is presented in Section~\ref{section:theory}. Our approach to optimisation of the trade-off between privacy and accuracy of black-box functions, and its contributions, are developed in Section~\ref{section:optimal}. A discussion of our work is contained in Section~\ref{section:discussion} and Section~\ref{section:conclusion} concludes the paper. 

\section{Related Works}
\label{section:related}
\paragraph{Information flow in programs}
Information flow analysis in imperative programs 
has been explored with many different approaches. One of the fundamental concepts is that of
security classes, introduced by Denning \cite{denning1976lattice}, which enables one to classify the variables of a program with respect to their level of confidentiality in order to form a lattice of information. Based on this classification, type systems \cite{volpano1996sound} and semantic approaches \cite{joshi2000semantic} have been implemented in order to define the security of instructions involving such variables. The most basic model considers only two security classes $L$ and $H$ separating the variables with a low and high level of confidentiality respectively \cite{denning1976lattice}. The security of a program is then expressed with the notion of non-interference between both classes \cite{volpano1996sound,smith2007principles,dima2006nondeterministic}. 
However, as programs in practice may contain some interference, other quantitative approaches \cite{yasuoka2014quantitative,clark2007static,clarkson2009quantifying,phan2014quantifying,malacaria2015algebraic,smith2011quantifying} have been proposed in order to measure the information flow that can arise between variables from different security classes.  
The computation of such quantitative information flows also includes the use of probabilistic instructions \cite{mciver2003probabilistic,joshi2000semantic,smith2007principles} that can randomise the algorithms and make programs non-deterministic and thus in some cases protect the confidentiality of information processed by variables in $H$. 

\paragraph{Information flow in SMC}
The security of SMC protocols ensures by definition that the participants can compute the public output of a public function of their private inputs without learning anything about the other parties' input, apart from what is inferable from the so called acceptable information flow we already discussed.

In \cite{ah2017secure}, we introduced a model of deceitful adversaries which enabled us to reason about the acceptable leakage, and to quantify, based on Shannon entropy, the information that such attackers can deduce from public outputs and their own private inputs. 
We also extended our model to a theoretic game that allows an attacker to evaluate the influence that he can have depending on the input he provides to the SMC protocol. 
In this present work, we build on this model to develop an approach that can mitigate or prevent this information leakage. We are able to do this for a large class of conditional entropies, which subsumes both the notions of R\'enyi entropy and $g$-entropy. 
We then introduce the notion of an approximate function, a corresponding non-linear optimisation problem, and we show how solving such optimisation problems can address certain privacy concerns raised in \cite{ah2017secure}.

\paragraph{Differential privacy}
The principle of randomising the output of a computation in order to protect the privacy of the data on which some calculations are performed is related to the concept of Differential Privacy (DP) \cite{dwork2008differential,dwork2014algorithmic}. 
DP formalises privacy concerns and introduces techniques that provide users of a database with the assurance that their personal details will not have a significant impact on the output of the queries performed on the database. 
More precisely, it proposes mechanisms which ensure that the outcome of the queries performed on two databases differing in at most one element will be statistically indistinguishable. 
Thus, DP means to reassure users of a database that their confidential data may be used in statistical measures 
without harming their individual privacy. 
Moreover, minimising the distortion of the outcome of the queries while ensuring privacy is an important trade-off that governs DP. 

Our approach aims at introducing concepts and mechanisms that can reassure participants of an SMC that they can engage in a computation whose public output will not affect their own privacy. Our privacy-enhancing techniques will also consider the accuracy of the randomised computation and they are akin to the Laplace mechanism in DP which blurs queries' outputs with  an additive noise. 
However, our aims and the methods we develop and use are significantly different. 
In DP, an attacker would take advantage of the information gained by repeating the same queries on two neighbouring databases, which would not be realistic in many applications of SMC such as in auctions or in e-voting. 
Instead, we focus on the amount of information on the private inputs that would flow from a single SMC. 
DP also does not model knowledge or belief of attackers whereas we model belief about protocol inputs of parties. 
Moreover, in our attack model a set of agents may collude in order to learn private information.
Adapting DP techniques to SMC would therefore not necessarily be always possible,  appropriate, nor even yield intended privacy guarantees.

\section{Background and Notations}
\label{section:background}

We recall different notions of entropy  used for quantifying information.

\ourparagraph{Notations} 
Let $D$ be a discrete set. We write 
$\cP(D)$ for the power set of $D$, and $\card{D}$ for the cardinality of set $D$.
Let $\Omega(D)$ be the set of all probability distributions whose support is contained in $D$. 
Throughout, we present distributions as Python dictionaries with domain values as keys and associated probabilities as values. For example, $\{4\colon\sfrac{1}{2}, 8\colon\sfrac{1}{2}\}$ represents the uniform distribution over $\{4, 8\}$. 
For any integers $a$ and $b$, we will write $\llrr{a}{b}$ for the set of consecutive integers ranging from $a$ to $b$, namely $\{a, a+1, \cdots, b\}$. The set of positive integers will be denoted by $\Np$ while $\Rp$ and $\Rpe$ will denote the set of non-negative and positive real numbers respectively. 
Let $n$ be in $\Np$. 
A \emph{linear distribution} over $\llrr{1}{n}$ will refer to the triangular distribution with mode $n$, i.e. to the distribution
$\{k \colon \frac{2k}{n(n+1)} \mid 1 \leq k \leq n\}$
where $\frac{2}{n(n+1)}$ is a normalising factor. 
Given random variable $X$ and value $x$, the event ``$X = x$'' will be abbreviated by ``$x$'' when there is no ambiguity, and its probability will be denoted by $p(x)$. 
Similarly, we will abbreviate $\sum_{x \in D}$ by $\sum_x$ when the domain $D$ is obvious from context. 

We denote by $\vv{x_i}_{1 \leq i \leq n}$ the $n$-dimensional vector in $\R^n$ whose coordinates are $x_1, \cdots, x_n$; we abbreviate this by $\vv{x_i}_i$ when there is no ambiguity. 
For all vector $v$ in $\R^n$ and all $p$ in $\Rp$, the usual $p$-norm of $v$ is denoted by $\norm{v}_p$. 
Let $\log$ denote the logarithm in base $2$ and let $\mu$ be the function defined for all non-negative real $x$ as:
\begin{equation}
\label{equ:mu}
\mu(x) = 
  \begin{cases}
  -x\cdot \log(x) & \text{if } x > 0 \\
  0 & \text{if } x = 0 
  \end{cases}
\end{equation}

\ourparagraph{Shannon and min-entropy}
Recall that for two random variables $X$ and $Y$ taking values in $\mathcal{X}$ and $\mathcal{Y}$, respectively, the \emph{Shannon entropy} $\HH_1(X)$ \cite{BLTJ:BLTJ1338} of $X$ and the Shannon entropy $\HH_1(X\mid Y)$ of $X$ given $Y$ are given as:
\begin{eqnarray}
\HH_1(X) = - \sum_x p(x) \log p(x)\label{equ:shannon} &{}&\qquad
\HH_1(X \mid Y) =\sum_y p(y) \HH_1(X \mid y)\label{equ:condshannon}
\end{eqnarray}

\noindent where \( \HH_1(X \mid y) = - \sum_x p(x \mid y) \log p(x \mid y) \). 
On the other hand, the \emph{Bayes vulnerability}  $\V_\infty (X)$ \cite{smith2011quantifying,cachin1997entropy,smith2009foundations}  expresses the probability of guessing a secret in one try. Similarly, the \emph{conditional} Bayes vulnerability \cite{dodis2004fuzzy} $\V_\infty(X \mid Y)$ of $X$ given $Y$ reflects the average probability of guessing the secret $X$ in one try. They are defined as:
\begin{eqnarray}
\V_\infty (X) = \max_x p(x) \label{equ:bayesvul} &{}&\qquad
\V_\infty(X \mid Y) = \sum_y p(y) \V_\infty(X \mid y) \label{equ:condbayesvul}
\end{eqnarray}

\noindent where
\( \V_\infty (X \mid y) = \max_x p(x \mid y) \). 
The 
\emph{min-entropy} of $X$ and \emph{conditional} min-entropy of $X$ given $Y$ are defined as:
\begin{eqnarray*}
\HH_\infty(X) = - \log \V_\infty(X) &{}& \qquad 
\HH_\infty(X \mid Y) = - \log \V_\infty(X \mid Y)
\end{eqnarray*}

		\ourparagraph{R\'enyi Entropy}
A more general notion of entropy, called \emph{R\'enyi entropy} \cite{renyi1961measures}, generalises both notions of Shannon entropy and min-entropy.  For sake of notational convenience, let us first define the \emph{$\alpha$-vulnerability} of $X$, for all positive real $\alpha \neq 1$, as:
\(\V_\alpha(X) = \norm{ \vv{p(x)}_x }_\alpha \)
Using this notion, we may express the \emph{R\'enyi entropy} $\HH_\alpha(X)$ of $X$ as:
\begin{eqnarray*}
\HH_\alpha(X) &=& \frac{\alpha}{1 - \alpha} \log \V_\alpha(X) 
\end{eqnarray*}

It is well-known and easily shown that the R\'enyi entropy $\HH_\alpha(X)$ converges towards the min-entropy $\HH_\infty(X)$ as $\alpha$ tends towards infinity. 
Moreover, an application of L'H\^opital's rule ensures that the R\'enyi entropy $\HH_\alpha(X)$ converges towards the Shannon entropy $\HH_1(X)$ as $\alpha$ tends towards $1$. 

However, although different notions of conditional R\'enyi entropy have been proposed, none of them has yet been commonly accepted as \emph{the} conditional R\'enyi entropy \cite{fehr2014conditional}. 
Yet, one candidate seems to be particularly suitable for our needs: Arimoto's \cite{arimoto1977information} notion of conditional R\'enyi entropy not only satisfies the natural properties of chain rule ($\HH_\alpha(X \mid Y) \geq \HH_\alpha(XY) - \log \card{\mathcal{Y}} $, where $\HH_\alpha(XY)$ denotes the joint entropy of $X$ and $Y$) and monotonicity ($\HH_\alpha(X \mid Y) \leq \HH_\alpha(X)$). But it is also compatible with both the Shannon entropy and the min-entropy in that we have the convergences  
$\HH_\alpha(X \mid Y) \xrightarrow[\alpha \to 1]{} \HH_1(X \mid Y)$ and
$\HH_\alpha(X \mid Y) \xrightarrow[\alpha \to \infty]{} \HH_\infty(X \mid Y)$. 
Therefore, we will introduce and work with the notion of conditional R\'enyi entropy due to Arimoto \cite{arimoto1977information}.
For sake of notational consistency, let us define the $\alpha$-vulnerability of $X$ given $Y$ as:
\(\V_\alpha(X \mid Y) = \sum_y p(y) \V_\alpha(X \mid y)\)
where 
\( \V_\alpha(X \mid y) = \norm{ \vv{p(x \mid y)}_x }_\alpha  \). 
For $\alpha\not=1$, we may now define  the \emph{conditional R\'enyi entropy} of $X$ given $Y$ as:
\begin{eqnarray*}
\HH_\alpha(X \mid Y) &=& \frac{\alpha}{1 - \alpha} \log \V_\alpha(X \mid Y) 
\end{eqnarray*}

		\ourparagraph{$g$-entropy}
The $g$-entropy \cite{m2012measuring}
measures the gain that someone might get from guessing a secret --- in our case, the private inputs of other parties. Since this is a relevant way of measuring risk of privacy violations, we wish that our approach and developed methods also support use of this notion of entropy.

Let $\mathcal{X}$ be the domain of $X$, where random variable $X$ models a \emph{secret}. Let $\mathcal{W}$ be the set of possible guesses for the value of $X$. A function $g$ of type  $\mathcal{W} \times \mathcal{X} \to [0, 1]$ is then called a \emph{gain function}. This function assigns to each guess $w$ in $\mathcal{W}$ and possible value $x$ of the secret in $\mathcal{X}$ a reward $g(w, x)$ that an attacker would gain by guessing $w$ when the secret value actually is $x$. Set $\mathcal{W}$ may be designed so that its elements refer to properties of secrets, values that are ``close'' to the secret or other means of expressing aspects of the secret.

For such a gain function $g$, the \emph{$g$-vulnerability of $X$}, $\V_g(X)$, is the expected reward that an attacker would gain by selecting his best guess. It, and the \emph{conditional} $g$-vulnerability $\V_g(X \mid Y) $ of $X$ given $Y$ are defined as:
\begin{eqnarray}
\V_g(X) = \max_w \sum_x p(x) g(w, x)\label{equ:gainvul} &{}&\qquad
\V_g(X \mid Y) = \sum_y \V_g(X \mid y)\label{equ:condgainvul}
\end{eqnarray}

\noindent where \( \V_g(X \mid y) = \max_w \sum_x p(x \mid y) g(w, x) \).
The \emph{$g$-entropy} and \emph{conditional $g$-entropy} are then defined as follows:
\begin{eqnarray*}
\HH_g(X) = - \log \V_g(X) &{}& \qquad
\HH_g(X \mid Y) = - \log \V_g(X \mid Y)
\end{eqnarray*}

The $g$-entropy generalises the min-entropy: for $\mathcal{W} = \mathcal{X}$ and gain function $id\colon \mathcal{X}\times \mathcal{X}\to [0,1]$~---~where $id(w,x) = 0$ if $w$ is not equal to $x$, and $id(x,x) = 1$ for all $x$ in $\mathcal{X}$, then $\HH_{\id}(X)$ equals $\HH_\infty(X)$ and $\HH_{\id}(X \mid Y)$ equals $\HH_\infty(X \mid Y)$.

\section{Generalised Conditional Entropy}
\label{section:entropy}
To get a very general definition of information leakage in Secure Multi-Party Computations (SMC), 
we define a more general notion of entropy that subsumes both R\'enyi entropy and $g$-entropy. 
For random variables $X$ and $Y$ with finite domain $\mathcal{X}$ and $\mathcal{Y}$, and a finite set $\mathcal{W}$ of possible guesses for $X$, we adapt the existing notions, such as the Bayesian vulnerability, to the presence of a gain function $g$ and its set of guesses $\mathcal{W}$. We indicate that dependency by writing $V_{\alpha,g}$ and so forth, subsequently.
We define properties of gain functions that are pertinent to our technical development.
\begin{definition}
\label{def:g_functions}
Let $g \colon \mathcal{W} \times \mathcal{X} \to [0, 1]$ be a gain function. 
\begin{enumerate}
\item
The gain function $g$ is \emph{positive} iff
\(\forall x \in \mathcal{X} \colon \sum_w g(w, x) > 0\)

\item
Let $\beta$ be in $\Rpe$. The gain function $g$ is $\beta$-positive iff
\(\forall x \in \mathcal{X} \colon \sum_w g(w, x) \geq \beta\)

\item The gain function $g$ is \emph{unitary} iff
\(\forall x \in \mathcal{X} \colon \sum_w g(w, x) = 1\)
\end{enumerate}
\end{definition}

We will only consider positive gain functions: a gain function that is not positive is the constant $0$ function, and can only produce $0$ vulnerabilities~---~as mentioned in \cite{m2012measuring}; $\beta$-positive gain functions will be useful in later sections. Note that, since $\mathcal{X}$ is finite, all positive gain function $g$ have some $\beta >0$ such that $g$ is $\beta$-positive.

Let $g \colon \mathcal{W} \times \mathcal{X} \to [0, 1]$ be a gain function and $0 < \alpha\not=1$.
The $(\alpha, g)$-vulnerability $\V_{\alpha,g} (X)$ of $X$ and the conditional $(\alpha, g)$-vulnerability $\V_{\alpha,g} (X \mid Y)$ of $X$ given $Y$ are defined as:
\begin{eqnarray}
\V_{\alpha,g} (X) \coloneqq \norm{ \vv{ \sum_x p(x) g(w, x) }_w }_\alpha\label{equ:alphavul} &{}&\qquad
\V_{\alpha,g} (X \mid Y) \coloneqq \sum_y p(y) \V_{\alpha,g} (X \mid y)\label{equ:condalphavul} 
\end{eqnarray}

\noindent where 
\(\V_{\alpha,g} (X \mid y) \coloneqq \norm{ \vv{ \sum_x p(x \mid y) g(w, x) }_w }_\alpha\).
We now define the $(\alpha, g)$-entropy of $X$ and the conditional $(\alpha, g)$-entropy of $X$ as:
\begin{eqnarray}
\HH_{\alpha,g}(X) \coloneqq	 
\frac{\alpha}{1-\alpha} \log \V_{\alpha,g}(X) &{}& \qquad
\HH_{\alpha,g}(X \mid Y) \coloneqq	 
\frac{\alpha}{1-\alpha} \log \V_{\alpha,g}(X \mid Y) \label{eq:general_entropy}
\end{eqnarray}

Again, we can easily verify that the $(\alpha, g)$-entropies $\HH_{\alpha,g}(X)$ and $\HH_{\alpha,g}(X \mid Y)$ both converge towards their respective $g$-entropies as $\alpha$ tends towards infinity.  We may thus define:
\begin{eqnarray*}
\HH_{\infty,g}(X) \coloneqq \HH_{g}(X) &{}& \qquad
\HH_{\infty,g}(X \mid Y) \coloneqq \HH_{g}(X \mid Y)
\end{eqnarray*}

We now focus on the case when $\alpha$ tends towards $1$ and we define, where $\mu$ is as in~(\ref{equ:mu}):
\begin{eqnarray*}
\HH_{1,g}(X) \coloneqq	
	\sum_w \mu \left( \sum_x p(x) g(w, x) \right) &{}&\qquad
\HH_{1,g}(X \mid Y) \coloneqq	
	\sum_w \mu \left( \sum_x p(x \mid y) g(w, x) \right)
\end{eqnarray*}

\noindent For \emph{unitary} gain functions $g$, it is easy to see that
the $(\alpha,g)$-entropies $\HH_{\alpha,g}(X)$ and $\HH_{\alpha,g}(X \mid Y)$ converge towards $\HH_{1,g}(X)$ and $\HH_{1,g}(X \mid Y)$, respectively, when $\alpha$ tends towards $1$. 
The reason for this is that,
when $g$ is unitary, the $(\alpha, g)$-vulnerabilities $\V_{\alpha,g}(X)$ and $\V_{\alpha,g}(X \mid Y)$ converge towards $0$ as $\alpha$ tends towards $1$. And then the claimed results follow from the application of L'H\^opital's rule, in a similar fashion as done for R\'enyi entropies. 
Let us formalise this:
\begin{lemma}
\label{lemma:hlim}
\begin{enumerate}
\item
Let $g \colon \mathcal{W} \times \mathcal{X} \to [0, 1]$ be a gain function. 
Then $\HH_{\alpha,g}(X)$ and $\HH_{\alpha,g}(X \mid Y)$ converge for $\alpha\to \infty$:	
\begin{eqnarray*}
\lim_{\alpha \to \infty} \HH_{\alpha,g}(X) = \HH_{\infty,g}(X) &{}&\qquad
\lim_{\alpha \to \infty} \HH_{\alpha,g}(X \mid Y) = \HH_{\infty,g}(X \mid Y)
\end{eqnarray*}

\item Moreover, if $g$ is unitary, then $\HH_{\alpha,g}(X)$ and $\HH_{\alpha,g}(X \mid Y)$ converge when $\alpha$ tends towards $1$, and we then have:
\begin{eqnarray*}
\lim_{\alpha \to 1} \HH_{\alpha,g}(X) = \HH_{1,g}(X) &{}&\qquad
\lim_{\alpha \to 1} \HH_{\alpha,g}(X \mid Y) = \HH_{1,g}(X \mid Y)
\end{eqnarray*}
\end{enumerate}
\end{lemma}

When the gain function $g$ is $\id$ with $\mathcal{W} = \mathcal{X}$ as above,  we obtain that for all positive reals $\alpha$, the $(\alpha, \id)$-entropies agree with the R\'enyi entropies:
\begin{eqnarray*}
\HH_{\alpha, \id}(X) = \HH_\alpha(X) &{}&
\qquad
\HH_{\alpha, \id}(X \mid Y) = \HH_\alpha(X \mid Y)
\end{eqnarray*}

\noindent This result is immediate for all values of $\alpha$ different from $1$. When $\alpha$ is equal to $1$, this follows from the fact that $\id$ is a unitary gain function and that we can apply the previous result ensuring that when $\alpha$ tends towards $1$, the $(\alpha, \id)$-entropies $\HH_{\alpha, \id}(X)$ and $\HH_{\alpha, \id}(X\mid Y)$, respectively, converge towards $\HH_{1, \id}(X)$ (the Shannon entropy) and $\HH_{1, \id}(X \mid Y)$ (the conditional Shannon entropy), respectively. We summarise those results and our discussion in Figure \ref{fig:arrayRenyi}.
\begin{figure}
\centering
\begin{tabular}{|l|l|l|l|} 
\hline
\diagbox[width=2.5cm, height=0.6cm]{$g$}{$\alpha$} & $\alpha = 1$ & $\alpha = \infty$ & $\alpha \in \left]0, \infty \right]$\\ \hline
$g = \id$ & Shannon entropy & min-entropy & R\'enyi entropy \\ \hline
$g \in [0, 1]^{\cW \times \cX}$ & {\footnotesize{only for unitary $g$}} & $g$-leakage & \\ \hline
\end{tabular}
\caption{Summary of the different notions of entropy that our generalised measure of information flow $\HH_{\alpha,g}$ subsumes. }
\label{fig:arrayRenyi}
\end{figure}

In conclusion, our new notion of entropy subsumes both the $g$-entropy and the whole family of R\'enyi entropies, including the Shannon entropy and the min-entropy. 
Therefore, all results that we develop in this paper will also be valid for all the different entropies mentioned earlier.

\section{Information Flow for Secure Multi-Party Computation}
\label{section:model}

\subsection{Model for Information Flow}
	\label{sec:smcflow}
Let us recall the technical setting and the assumptions introduced in \cite{ah2017secure} for studying and quantifying the information flow produced by public outputs in SMC, as this constitutes a basis for the remaining technical developments in this paper.
Throughout this paper, we consider a set of $n> 1$ parties $\bP = \{P_1, \cdots, P_n\}$ holding the respective inputs $x_1, \cdots, x_n$, each of them belonging to $\Z$. 
Let $f\colon \Z^n \to \Z$ be a function. 
Let $o$ denote the output of the function applied with the parties' inputs, i.e. $o = f(x_1, \cdots, x_n)$. Both $o$ and $f$ are public and so known to all parties in $\bP$.
In order to study the aforementioned acceptable information leakage of this situation, we introduce the following model. 
Let $\bA$ and $\bT$ be two non-empty subsets of $\bP$ and $\bS$ be a possibly empty subset of $\bP$ such that $(\bA, \bT, \bS)$ forms a partition of $\bP$. 
Our attack models assumes that 
all parties in $\bA$ are willing to collaborate between each other in order to maximise information leakage on inputs of the parties in $\bT$. 
The sets $\bA$, $\bT$ and $\bS$ will thus respectively be referred to as the sets of \emph{attackers, targets and spectators}, respectively. 
We now define the attackers' input $\xA = \vv{x_i}_{i \in \bA}$, the targets' input $\xT = \vv{x_i}_{i \in \bT}$ and the spectators' input $\xS = \vv{x_i}_{i \in \bS}$. By abuse of notation (or a reordering of arguments for $f$), we will also refer to the output specification of $f$ as $o = f(\xA, \xT, \xS)$. 

Let $a = \card{\bA}$, $t = \card{\bT}$ and $s = \card{\bS}$ denote the cardinality of the respective sets.
Let $\DA$ be an element of $ \mathcal{P}(\Z)^a$ and let us assume that the input vector of the parties in $\bA$ is ranged in $\DA$. 
Similarly, let $\DT$ in $\mathcal{P}(\Z)^t$ and $\DS$ in $\mathcal{P}(\Z)^s$ be the domain of the input vectors of the parties in $\bT$ and $\bS$ respectively. 
In other words, we assume that:
\[
\xA \in \DA, \qquad \xT \in \DT, \qquad \xS \in \DS
\]

\noindent However, as those inputs are private, their exact value is not known to the other parties. 
In order to quantify the information leaks that output $o$ produces, we model the parties' inputs as random variables $\XA$, $\XT$ and $\XS$ respectively, following the respective probability distributions:
\[
\pA \in \Omega(\DA), \qquad \pT \in \Omega(\DT), \qquad \pS \in \Omega(\DS)
\]

\noindent where $\Omega(X)$ is the set of discrete probability distributions over a finite set $X$. These probability distributions will model the beliefs that each set of parties has on the other parties' inputs. More precisely, the parties in $\bA$ and $\bT$ believe that random variable $\XS$ is governed by $\pS$, the parties in $\bA$ and $\bS$ believe that $\XT$ follows $\pT$, whereas the parties in $\bT$ and $\bS$ believe that $\XA$ follows $\pA$. We articulate the assumptions we make about these distributions:
\begin{assumption}
\label{ass:beliefs}
We assume that the parties' beliefs $\pA$, $\pT$ and $\pS$ are public and are part of the common knowledge amongst all parties in $\bP$.
Moreover, our model assumes that the three groups of parties will not collaborate between each other and that their inputs are thus independent. 
\end{assumption}

The independence of $\XA$, $\XT$ and $\XS$ will play an important role in the proofs of the Theorems in Section \ref{section:theory}. 
The assumption that their probability distributions are public and part of the common knowledge ensures that all the parties will be able to access the same data produced by our measure of information flow in Section \ref{sec:awaes} and Section \ref{section:random}, and will be able to reach a consensus regarding how to best protect the targeted inputs' privacy, as discussed in Section \ref{section:optimal}. 
These probability distributions can express a variety of beliefs from uniform to point mass distributions. 

Lastly, let $\DO$ in $\mathcal{P}(\Z)$ be the output domain, defined as $\DO = \{ f(\xA, \xT, \xS)\mid \xA\in\DA, \xT\in\DT, \xS\in \DS\}$. As a function of random variables, the output $o = f(\xA, \xT, \xS)$ will therefore be modelled by the random variable:
\begin{equation}
\label{equ:Of}
O_f = f(\XA, \XT, \XS)
\end{equation}

\noindent ranged in $\DO$. We sometimes write $O$ when $f$ is clear from context.
In order to quantify the information that the attackers would learn about $\XT$ when inputting a particular input $\xA$, we introduced in \cite{ah2017secure} the attackers' weighted average entropy $\awae{\bT}{\bA}$ defined for all $\xA$ in $\DA$ as the conditional \emph{Shannon} entropy of $\XT$ given $O$ and $\xA$, i.e.:
\begin{equation}
\label{equ:awaeShannon}
\awae{\bT}{\bA}(\xA) = \sum_o p(o \mid \xA) \sum_{\xT} \mu( p(\xT \mid o, \xA) )
\end{equation}
where $\mu$ was defined in~(\ref{equ:mu}).

A deceitful attacker, i.e. an attacker who is willing to lie on his honest and intended input in order to learn more information on the private inputs of his targets, will now be able to take advantage of this indicator in~(\ref{equ:awaeShannon}) in order to shape his input so as to maximise his information gain. 
Since the notion of $\awae{\bT}{\bA}$ in (\ref{equ:awaeShannon}) is an instance of the conditional Shannon entropy, we need to widen the approach and analyses of~\cite{ah2017secure} to make them compatible with more general notions of entropy. We develop this next.

\subsection{General Attackers' Entropy}
\label{sec:awaes}
Function $\awae{\bT}{\bA}$ for measuring information leakage is dependent on some implicit parameters, namely the SMC function $f$, the partition $(\bA, \bT, \bS)$ of $\bP$ and the distributions $\pT$ and $\pS$ of the targets and spectators' inputs. 
Our technical development
needs to make those parameters explicit, and it needs to work for the generalised entropy notion presented in Section \ref{section:entropy}. 
Therefore, we now define a
higher-order function $\Awae$ which fulfils those requirements. 
Subsequently, we will work with a 
set of allowable guesses $\cW$ for the targeted input $\xT$. 
\begin{definition}
\label{def:awaes}
Let $\alpha$ be in $\Rpe \cup \{ \infty \}$ and $g\colon \cW \times \DT \to [0,1]$ be a gain function. 
We introduce the higher-order function $\Awae_{\alpha,g}$ of type:
\[\Awae_{\alpha,g} \colon (\Z^n\to \Z) \times {\cP(\bP)}^3 \times \Omega(\Z)^2 \to (\DA\to \R^+)\]

\noindent that takes as arguments an SMC function $f$ of type $\Z^n\to \Z$, three disjoint sets of participants $(\bA, \bT, \bS)$ that form a partition of $\bP$, the probability distribution $(\pT, \pS)$ of the respective targets' and spectators' inputs, and returns a function $\Awae_{\alpha,g}(f, (\bA, \bT, \bS), (\pT, \pS))$ of type $\DA \to \Rp$, denoted as $\awae{\alpha,g}{f}$ and defined for all $\xA$ in $\DA$ as the conditional $(\alpha, g)$-entropy of $\XT$ given $O_f$ as in~(\ref{equ:Of}) and $\xA$:
\begin{eqnarray*}
\awae{\alpha,g}{f}(\xA) &=&
\HH_{\alpha,g}(\XT \mid O_f, \xA) 
\end{eqnarray*}

\end{definition}

For subsequent theorems and proofs, we note that for $0< \alpha\not=1$ we have:
\begin{equation}
\label{equ:positivealphafact}
\awae{\alpha, g}{f}(\xA)
= \frac{\alpha}{1-\alpha} \cdot \log \V_{\alpha,g}(\XT \mid O_f, \xA)
\end{equation}

\noindent where the $(\alpha, g)$-vulnerability $\V_{\alpha,g}(\XT \mid O_f, \xA)$ can be written as:
\begin{eqnarray}
\V_{\alpha,g}(\XT \mid O, \xA)
&=&
\sum_o p(o \mid \xA) \cdot \norm{ \vv{ \sum_\xT p(\xT \mid o, \xA) \cdot g(w, \xT) }_w }_\alpha \\
&=& 
\sum_o \norm{ \vv{ \sum_\xT p(\xT) \cdot p(o \mid \xT, \xA) \cdot g(w, \xT) }_w }_\alpha\label{equ:Vag}
\end{eqnarray}

This is so since $\alpha$-norm is homogeneous, even for $\alpha < 1$, and since $\xA$ and $\xT$ are independent random variables. 

This new function $\awae{\alpha, g}{f}$ provides us with a generic way of measuring information flow. 
Indeed, it subsumes some notions of entropy that are widely used in cryptography. For example, when $g$ equals $\id$, this function corresponds to the conditional R\'enyi entropy. When $\alpha$ equals $\infty$, it corresponds to the conditional $g$-entropy. 
We also observe that when $\alpha$ equals $1$ and $g$ equals $\id$, our new function $\awae{1, \id}{f}$ is identical to the function $\awae{\bT}{\bA}$ introduced in \cite{ah2017secure}.

We now illustrate how our general measure of information flow in Secure Multi-Party Computations enables us to quantify the information that attackers can gain on their targets' inputs. 
In doing so, we also raise interesting concerns that will further motivate our present work. Let us consider an example.
\begin{example}
\label{ex:basic_awae}
Let us consider $3$ parties $X$, $Y$ and $Z$ holding the respective inputs $x$, $y$ and $z$, and where $\bA = \{X\}$ is attacking $\bT = \{Y\}$ under spectator $\bS = \{Z\}$. 
Let $\DA = \DT = \DS = \llrr{1}{30}$ and let us assume that $\XT$ and $\XS$ are uniformly distributed over this domain. 
Let $f \colon \Z^3 \to \Z$ be defined by $f(x, y, z) = x(2y + z) + 2z$. 

In this example, we will study the behaviour of the conditional min-entropy of the targeted inputs. 
In other words, we will instantiate $\alpha$ with $\infty$ and $g$ with $\id$ in order to study the function $\awae{\infty,\id}{f}$ which we plot in Figure \ref{fig:basic_awae}. This plot clearly shows that some values of $\xA$ are more advantageous for attacker $X$ in that they produce lower conditional entropies for his targeted input $Y$. 
For instance, inputting $x=2$ would produce a high entropy and would not reveal much information about $y$. 
In contrast, input $x=15$ would produce entropy $0$, which means that $X$ would learn the exact value of $y$ from the output. 
Indeed, as $X$ knows his own input, he knows that in this case, the output equals $o = f(15, y, z) = 30y + 17z$. 
We can check that for all $z$ in $\DS$ the function $f_z \colon y \mapsto f(15, y, z)$ is bijective from $\DT$ to $f_z(\DT)$ as both sets have size $30$. 
This thus ensures that attacker $X$ can deduce the exact value of $y$ when learning the output value. 
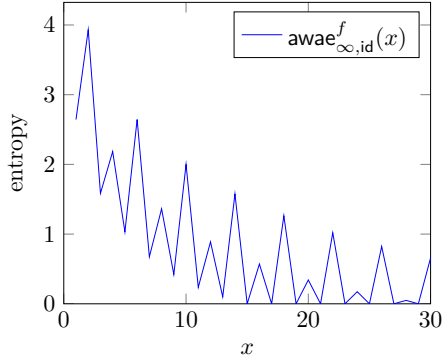
\begin{figure}
\centering
\begin{tikzpicture}[scale=.9]
	\begin{axis}[
	scale=1,
	 ymin=0,
	 xmin=0,
	 xmax=30,
	  xlabel=$x$,
	  ylabel=entropy,
	  legend pos=north east]
	\addplot[blue] table [y=fx, x=x, mark=none]{minBounds.dat};
	\addlegendentry{$\awae{\infty, \id}{f}(x)$}
\end{axis}
\end{tikzpicture}
\caption{Illustration of $\awae{\infty,\id}{f}$ in the computation of function $f(x, y, z) = x(2y + z) + 2z$ with $\pT$ and $\pS$ uniform over $\llrr{1}{30}$, when $X$ attacks $Y$ under spectator $Z$. 
}
\label{fig:basic_awae}
\end{figure}
\end{example}

We just saw that the choice of the attackers' input $\xA$ can have a dramatic influence on the entropy of the targeted input $\xT$. 
In particular, the attackers can harm the privacy of their targets by choosing some judicious inputs $\xA$. 
In order to mitigate against this privacy concern, we next introduce and study the notion of \emph{approximate SMC computation}.

\section{Function Randomisation Via Virtual Inputs}
\label{section:random}

We now consider the case where revealing the exact value of the output of $f$, namely $o = f(\xA, \xT, \xS)$, would be likely to jeopardise the privacy of the targeted input $\xT$. Thus, we would like to be able to replace the computation of $f$ by the computation of an approximate function $f'$, whose output should not only be a decent indicator of $o$, but should also enhance the privacy of $\bT$'s input. 
This presents an inherent trade-off between the accuracy of the output and the privacy of the inputs. We will understand this trade-off in detail in Section \ref{section:optimal}. 

In order to randomise the observed output, 
the function $f'$ will take an additional argument $\phi$, that may consist in a number of integral inputs, and that will act as a source of randomness that can distort the output to protect privacy of targeted inputs. 
Let us next formalise this notion of approximate function. 
\begin{definition}
\label{def:approximate}
Let $n$ be in $\Np$, $v$ be in $\N$ and $f \colon \Z^n \to \Z$ be an $n$-ary function. 
\begin{enumerate}

\item Function $f' \colon \Z^n \times \Z^v \to \Z$ is an \emph{approximation} of $f$ or that $f'$ \emph{approximates} $f$ iff there exists a function $h \colon \Z \times \Z^v \to \Z$ such that:
\begin{equation}
\forall \xx \in \Z^n, 
\forall \phi \in \Z^v \colon
f'(\xx, \phi) = h(f(\xx), \phi)
\end{equation}

\item An approximation $f'$ of $f$ is a \emph{close approximation} of $f$ --- or $f'$ \emph{closely approximates} $f$ --- iff for all $\phi$ in $\Z^v$, the function $h_{\phi} \colon \Z \to \Z$ is injective, where $h_\phi$ is defined for all $\phi$ in $\Z^v$ as
\(\forall o \in \Z \colon
h_{\phi}(o) = h(o, \phi)\).

\item We define $\tid{f} \colon \Z^{n+1} \to \Z$, the \emph{additive approximation} of $f$, for all $\xx$ in $\Z^n$ and $\phi$ in $\Z$ as 
\(\tid{f}(\xx, \phi) = f(\xx) + \phi.\)
\end{enumerate}

\end{definition}

We illustrate the notion of approximate function $f'$ for a function $f$ in Figure \ref{fig:bb_model}.  Function $f'$ has all inputs of $f$ and additional virtual inputs $\varphi$; and its black box contains ``internal wirings'' so that $\varphi$ and the output $o$ of $f$ are fed into function $h$ within that black box to produce approximate output $o'$. A \emph{close} approximation $f'$ of $f$ requires all the functions $h_\phi$ to be injective, which makes sense for SMC as it enforces a correlation between the output of $f$ and that of its approximation $f'$. 
Indeed, knowledge of $o'$ and $\phi$ determine that of $o$, which prevents $o'$ to be independent from $o$. 
We also note that the \emph{additive} approximation $\tid{f}$ of a function $f$ closely approximates the latter. 
\begin{figure}
\centering
\includegraphics[scale=1]{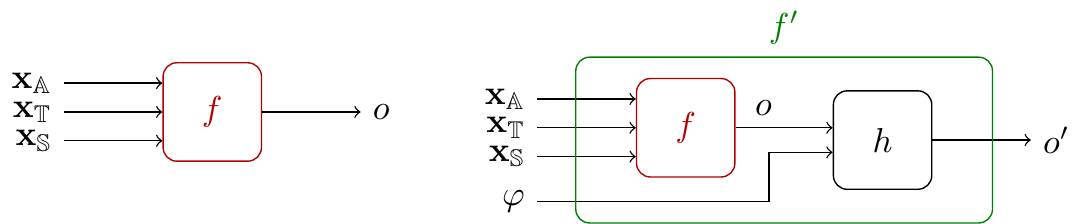}
\caption{Comparison of the black-box model for SMC of function $f$ (left) with that of its approximation $f'$ (right), as introduced in Definition \ref{def:approximate}. The virtual inputs $\varphi$ and the output $o$ of $f$ are fed into function $h$ within the black box to produce approximate output $o'$.}
\label{fig:bb_model}
\end{figure}

The use of a substitute function $f'$ aims to contain and limit the information that would flow from the computation of $f$ by randomising the output of $f$ with an additional variable $\phi$. 
Therefore, we need to understand and quantify the information flow that the computation of such an approximate function $f'$ produces, and we need to study and represent the behaviour of the additional variable $\phi$ that $f'$ uses to randomise the output of $f$. To ensure the security of such approximations, variable $\phi$ is not held by any physical party, it is a \emph{virtual input}, a concept we formalise next. 

\begin{definition}
Let $n$ and $v$ be in $\Np$. 
A $v$-dimensional \emph{virtual input} $\phi$ is a vector in $\Z^v$, \emph{independent from the other inputs}, and not held by any party in $\bP$. 
As such, its value $\phi$ is kept secret and appears to all the parties in $\bP$ as a random variable $\Phi$ on domain $\DP$ following a probability distribution $\pP$, referred to as the \emph{virtual distribution}. 
A set of \emph{virtual parties} $\bV$ is deemed to be the (virtual) owner of $\phi$. 
\end{definition}

\noindent In other words, the probability distribution $\pP$ can be regarded as the prior belief that all the parties in $\bP$ have on input $\phi$. 
Note that all those parties in $\bP$ will have the same public prior belief on $\phi$, in accordance with Assumption \ref{ass:beliefs}, and that $\bP$ and $\bV$ are mutually disjoint. 

The set of parties $\bP'$ for function $f'$ is
\(\bP' = \bP \cup \bV \). 
We now study the privacy that targeted parties gain when the computation of a function $f$ is substituted for that of an approximation $f'$, randomised by a virtual input $\phi$. 
\begin{definition}
Let $n> 1$ and $v$ in $\Np$. 
Let $f \colon \Z^n \to \Z$ be a function and let $f' \colon \Z^n \times \Z^v \to \Z$ approximate $f$. 
Let a virtual input $\phi$ be in $\Z^v$ and let $\pP$ be its probability distribution. 
Finally, let $\alpha$ be in $\R \cup \{ \infty \}$ and $g$ be a gain function of type $\cW \times \DT\to [0,1]$. 
Using the joint probability distribution defined by $(\pS \cdot \pP)(\xS, \phi) \coloneqq \pS(\xS) \cdot \pP(\phi)$ for all $\xS$ in $\DS$ and $\phi$ in $\DP$, function $\awae{\alpha, g}{f', \pP} \colon \DA \to \Rp$ is given as:
\begin{equation}
\awae{\alpha, g}{f', \pP} \coloneqq \Awae_{\alpha,g}(f', (\bA, \bT, \bS \cup \bV), (\pT, \pS \cdot \pP))
\end{equation}

\end{definition}

This function $\awae{\alpha, g}{f', \pP}$ measures the privacy of the targets, given a certain approximate function and virtual input distribution. 
It will be particularly useful, for studying how privacy changes for
different virtual input distributions. 
The assumption that for $f'$ and $f$, the sets $\bA$ and $\bT$ are unchanged, does not compromise the security of our approach: an attacker for function $f'$ could not really learn anything useful about the input of parties in $\bV$, since these inputs are randomly drawn according to $\pi_{\Phi}$. 
Let us illustrate the benefits offered by function substitution. 
\begin{example}
\label{ex:rand_awae}
Let us re-consider the scenario of Example~\ref{example:basic}, but now with 
the additive approximation $\tid{f}$ of $f$. 
We will study the behaviour of the conditional min-entropy of the targeted inputs when we approximate $f$ with $\tid{f}$. In other words, we will study the function $\awae{\infty, \id}{\tid{f}, \pP_i}$ for the following distributions $\pP_i$:
\begin{eqnarray*}
\pP_1 &=& \{-2\colon \sfrac{1}{4}, 0\colon\sfrac{1}{4}, 2\colon\sfrac{1}{4}, 4\colon\sfrac{1}{4}\} \\
\pP_2 &=& \{-1\colon \sfrac{1}{4}, 0\colon\sfrac{1}{4}, 1\colon\sfrac{1}{4}, 2\colon\sfrac{1}{4}\} \\
\pP_3 &=& \{-3\colon \sfrac{1}{8}, -2\colon \sfrac{1}{8}, -1\colon\sfrac{1}{8}, 0\colon \sfrac{1}{4}, 1\colon\sfrac{1}{8}, 2\colon \sfrac{1}{8}, 3\colon\sfrac{1}{8}\}
\end{eqnarray*}

As seen in Figure~\ref{fig:rand_awae}, for all $1\leq i\leq 3$, function $\awae{\infty, \id}{\tid{f}, \pP_i}$ is above $\awae{\infty, \id}{f}$.  This suggests that randomising a computation effectively enhances the privacy of the targeted inputs. 
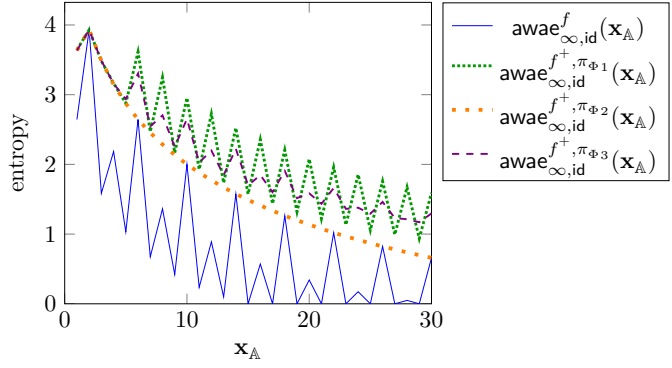
\begin{figure}
\centering
\begin{tikzpicture}[scale=.9]
	\begin{axis}[
	scale=1,
	 ymin=0,
	 xmin=0,
	 xmax=30,
	  xlabel=$\xA$,
	  ylabel=entropy,
	  legend pos=outer north east]
	\addplot[blue] table [y=fx, x=x, mark=none]{minBounds.dat};
	\addlegendentry{$\awae{\infty, \id}{f}(\xA)$}
	\addplot[green!60!black, densely dotted, very thick] table [y=p1, x=x, mark=none]{minBounds.dat};
	\addlegendentry{$\awae{\infty, \id}{\tid{f}, \pP_1}(\xA)$}
	\addplot[orange, loosely dotted, ultra thick] table [y=p2, x=x, mark=none]{minBounds.dat};
	\addlegendentry{$\awae{\infty, \id}{\tid{f}, \pP_2}(\xA)$}
	\addplot[violet, dashed, thick] table [y=p3, x=x, mark=none]{minBounds.dat};
	\addlegendentry{$\awae{\infty, \id}{\tid{f}, \pP_3}(\xA)$}
\end{axis}
\end{tikzpicture}
\caption{Benefits of substituting the computation of $f$ by that of its approximation $\protect\tid{f}$ in the computation of $f(x, y, z) = x(2y + z) + 2z$ with $\pT$ and $\pS$ uniform over $\llrr{1}{30}$. For all $1\leq i\leq 3$, the function $\awae{\infty, \id}{\protect\tid{f},\pP_i}$ is above $\awae{\infty, \id}{f}$. 
}
\label{fig:rand_awae}
\end{figure}
\end{example}

The latter example indicates that function randomisation indeed contributes to improving the privacy of the targets. 
In the next section, we want to formally investigate the privacy gains offered by function randomisation. 
In particular, we would like to understand why substituting the computation of a function $f$ by that of an approximation $f'$ can only enhance the privacy of the targets, and we will further characterise this privacy gain for \emph{close} approximations.

\section{Theory of Virtual Input Randomisation}
\label{section:theory}

We first summarise the mathematical setting we study in the remainder of this paper:
\begin{assumption}
In the remainder of this paper, including lemmas and theorems,
$f'$ is an approximation of $f$, where $\phi$ is a virtual input with domain $\DP$. Moreover, $g\colon \cW \times \DT \to [0,1]$ is a positive gain function, and 
$\beta$ is a positive real. 
\end{assumption}

The following theorem states that the computation of any approximate function $f'$ will not produce a lower privacy for the targeted inputs than that produced by the computation of $f$. 
\begin{theorem}
\label{thm:lowerbound}
Let $\alpha$ be a positive real different from $1$. 
Then, we have:
\begin{equation}
\forall \pP \in \Omega(\DP), 
\forall \xA \in \DA \colon
\awae{\alpha, g}{f', \pP}(\xA) \geq 
\awae{\alpha, g}{f}(\xA)
\end{equation}
\end{theorem}

\begin{proof}
Let $D \coloneqq \DA \times \DT \times \DS$ and $X \coloneqq (\XA, \XT, \XS)$. 
By definition, since $f'$ approximates $f$, there is a function $h$ such that $f'(\xx, \phi) = h(f(\xx), \phi)$ for all $\xx$ in $D$ and all $\phi$ in $\DP$. 
The random variable representing the output of $f$, namely $O = f(X)$ has domain $\DO$.  Similarly, let $\DOp$ be the domain of the output of $f'$, namely $O'=f'(X, \Phi) = h(f(X), \Phi) = h(O, \Phi)$. 
Let $\pP$ be in $\Omega(\DP)$ and $\xA$ be in $\DA$. We recall that we have:
\[
\awae{\alpha, g}{f', \pP}(\xA)
= \frac{\alpha}{1-\alpha} \cdot \log \V_{\alpha,g}(\XT \mid O', \xA)
\]

\noindent where:
\begin{equation}
\label{eq:vuln_proof}
\V_{\alpha,g}(\XT \mid O', \xA)
= \sum_{o'} \norm{ \vv{ \sum_\xT p(\xT) \cdot p(o' \mid \xT, \xA) \cdot g(w, \xT) }_w }_\alpha
\end{equation}

Applying Bayes Theorem twice, and as $\Phi$ is independent from $\XA$, $\XT$ and $O$, we obtain that:
\begin{eqnarray}
p(o' \mid \xT, \xA) &=& 
\label{eq:bayes_once}
\sum_{\phi} p(\phi) \cdot p(o' \mid \xA, \xT, \phi) \\ &=&
\sum_{\phi} p(\phi) \cdot \sum_{o \in h_{\phi}^{-1}(o')} p(o \mid \xA, \xT)\cdot  p(o' \mid \xA, \xT, \phi, o) \nonumber
\end{eqnarray}
since $p(o' \mid \xA, \xT, \phi, o)\not=0$ only when $o$ is in $h_{\phi}^{-1}(o')$. 
Moreover, $p(o' \mid \xA, \xT, \phi, o) = 1$ for $o$ in $h_{\phi}^{-1}(o')$. 

\ourparagraph{Case $\alpha > 1$}
We can apply the triangular inequality twice from Equation (\ref{eq:vuln_proof}) in order to obtain:
\begin{equation}
\label{eq:triangular}
\V_{\alpha,g}(\XT \mid O', \xA) \leq
\sum_{o'} \sum_{\phi} p(\phi) \sum_{o \in h_{\phi}^{-1}(o')} 
\norm{ \vv{ \sum_\xT p(\xT) \cdot p(o \mid \xA, \xT) \cdot g(w, \xT) }_w }_\alpha
\end{equation}

For any given $\phi$ in $\DP$, the collection of sets $(h^{-1}_\phi(o'))_{o' \in \DOp}$ constitutes a partition of $\DO$. 
So there exists a unique $o'$ in $\DOp$ such that $o$ is in $h^{-1}_\phi(o')$. 
We can thus simplify the double summation over $o'$ and $o$ as a single sum over $o$:
\[
\V_{\alpha,g}(\XT \mid O', \xA) \leq
\sum_{\phi} p(\phi) \cdot \sum_{o}
\norm{ \vv{ \sum_\xT p(\xT) \cdot p(o \mid \xA, \xT) \cdot g(w, \xT) }_w }_\alpha
\]

\noindent Since $\alpha$ is greater than $1$, the expression $\frac{\alpha}{1-\alpha}$ is negative and we get:
\begin{equation}
\label{eq:proof_conclu}
\awae{\alpha, g}{f', \pP}(\xA) \geq 
\awae{\alpha, g}{f}(\xA)
\end{equation}

\ourparagraph{Case $\alpha < 1$} 
We can show that for all $n$ in $\Np$, for all $x$ and $y$ in $(\Rp)^n$, we have $\norm{x+y}_\alpha \geq \norm{x}_\alpha + \norm{y}_\alpha$. This follows from Minkowski inequality in the case where $\alpha$ is lower than $1$, since $x \mapsto x^\alpha$ is then concave on $\Rp$. 
This reversed triangular inequality reverses the inequality obtained in (\ref{eq:triangular}) and as $\frac{\alpha}{1 - \alpha}$ is now positive, we find the same result as in (\ref{eq:proof_conclu}). 
\end{proof}

The proof of the previous theorem is based on the analysis of the formal expressions of $\awae{\alpha,g}{f', \pP}(\xA)$ and $\awae{\alpha, g}{f}(\xA)$ when $0< \alpha\not=1$. However, we can extend this result to $\alpha=1$ and $\alpha=\infty$,
by appealing to that result for positive $\alpha\not=1$ and the continuity of inequalities under limits:
\begin{corollary}
\label{cor:lowerbound}
\begin{enumerate}
\item
We have:
\(\forall \pP \in \Omega(\DP), 
\forall \xA \in \DA \colon
\awae{\infty, g}{f', \pP}(\xA) \geq 
\awae{\infty,g}{f}(\xA).\)

\item Moreover, if $g$ is unitary, then we have:
\(\forall \pP \in \Omega(\DP), 
\forall \xA \in \DA \colon
\awae{1, g}{f', \pP}(\xA) \geq 
\awae{1, g}{f}(\xA).\)
\end{enumerate}
\end{corollary}

\begin{proof}
By virtue of Lemma \ref{lemma:hlim}, we know that letting $\alpha$ tend towards $\infty$ in Theorem \ref{thm:lowerbound} yields the result stated in item 1) above. 
Similarly, if $g$ is unitary, Lemma \ref{lemma:hlim} ensures that Theorem \ref{thm:lowerbound} implies the result stated in item 2)
as $\alpha\to 1$. 
\end{proof}

Concretely, the theorem states that learning a function of the output of $f$ cannot leak more information on the inputs of $f$ than the output of $f$ may leak already. 
On the other hand, we are able to estimate an upper bound for the privacy of the inputs of the targeted parties, once a computation has been randomised. 
The next theorem states that, when replacing the computation of a function $f$ by that of a \emph{close} approximation $f'$, the entropy gain provided by a virtual input cannot exceed the entropy of the distribution for the virtual inputs. 
\begin{theorem}
\label{thm:upperBound}
Let $f'$ be a close approximation of $f$ and $0< \alpha\not=1$. 
Then, we have:
\begin{equation}
\forall \pP \in \Omega(\DP), 
\forall \xA \in \DA \colon
\awae{\alpha, g}{f', \pP}(\xA) \leq 
\awae{\alpha, g}{f}(\xA) + \HH_{\alpha}(\pP)
\end{equation}
where $\HH_{\alpha}(\pP)$ refers to the R\'enyi entropy of order $\alpha$ of the distribution $\pP$. 
\end{theorem}

\begin{proof}
By definition, since $f'$ closely approximates $f$, there exists some function $h$ such that $f'(\xx, \phi) = h(f(\xx), \phi)$ for all $\xx$ in $D$ and $\phi$ in $\DP$. 
Let $\pP$ be in $\Omega(\DP)$ and $\xA$ be in $\DA$. 
For sake of readability, we set
\(V' = \V_{\alpha,g}(\XT \mid O', \xA)
\)
and use $V'$ in the arguments below. From Equation (\ref{eq:bayes_once}), we recall that:
\begin{eqnarray*}
\V' &=& 
\sum_{o'} \left( \sum_w \left[ \sum_\xT p(\xT) \cdot
\sum_{\phi} p(\phi) \cdot p(o' \mid \xA, \xT, \phi) \cdot
g(w, \xT) \right]^\alpha \right)^{\frac{1}{\alpha}}
\end{eqnarray*}

\ourparagraph{Case $\alpha > 1$}
We know that $x \mapsto x^\alpha$ is convex on $\Rp$ and equals $0$ at $0$. We also know that $x \mapsto x^{\frac{1}{\alpha}}$ is increasing on $\Rp$ and thus:
\begin{eqnarray}
\V' &\geq&
\label{eq:first_ineq_up}
\sum_{o'} \left( \sum_w \sum_{\phi} \left[ \sum_\xT p(\xT) \cdot
p(\phi) \cdot p(o' \mid \xA, \xT, \phi) 
\cdot g(w, \xT) \right]^\alpha \right)^{\frac{1}{\alpha}} \\ 
&\geq&
\sum_{o'} \left( \sum_w \sum_{\phi} p(\phi)^\alpha\cdot 
\left[ \sum_\xT p(\xT)\cdot
p(o' \mid \xA, \xT, \phi) \cdot
g(w, \xT) \right]^\alpha \right)^{\frac{1}{\alpha}} \nonumber
\end{eqnarray}

Let us denote $\sum_{\phi} p(\phi)^\alpha$ by $\sigma$. 
For any $\phi$, we have
$p(\phi)^\alpha = \sigma \cdot \frac{p(\phi)^\alpha}{\sigma}$. But also $\sum_\phi \frac{p(\phi)^\alpha}{\sigma}$ equals $1$. 
We also know that $x \mapsto x^{\frac{1}{\alpha}}$ is concave. Therefore, Jensen's inequality yields:
\begin{eqnarray}
\label{eq:jensen_proof_up}
\V' &\geq&
\sigma^{\frac{1}{\alpha}}
\sum_{\phi} \frac{p(\phi)^\alpha}{\sigma}
\sum_{o'}
\left( \sum_w 
\left[ \sum_\xT p(\xT)\cdot
p(o' \mid \xA, \xT, \phi) \cdot
g(w, \xT) \right]^\alpha \right)^{\frac{1}{\alpha}}
\end{eqnarray}

Moreover, we have $p(o' \mid \xA, \xT, \phi) = p(O \in h_{\phi}^{-1}(o') \mid \xA, \xT)$ since $O$ and $\Phi$ are independent. 
Furthermore, for all $\phi$ in $\DP$, we know that $h_{\phi}$ is injective. Thus, from~(\ref{eq:jensen_proof_up}) we get that:
\begin{eqnarray*}
\V' &\geq&
\sigma^{\frac{1}{\alpha}}
\sum_{\phi} \frac{p(\phi)^\alpha}{\sigma}
\sum_{o}
\left( \sum_w 
\left[ \sum_\xT p(\xT)\cdot 
p(o \mid \xA, \xT) \cdot 
g(w, \xT) \right]^\alpha \right)^{\frac{1}{\alpha}} \\
&\geq&
\sigma^{\frac{1}{\alpha}}
\sum_{o}
\norm{ \vv{ 
\sum_\xT p(\xT)\cdot
p(o \mid \xA, \xT) \cdot
g(w, \xT) }_w }_\alpha
\end{eqnarray*}

\noindent and as $\frac{\alpha}{1-\alpha}$ is negative, the claim follows:
\begin{equation}
\label{eq:conclu_proof_up}
\awae{\alpha, g}{f', \pP}(\xA) \leq 
\awae{\alpha, g}{f}(\xA) + \HH_{\alpha}(\pP)
\end{equation}

\ourparagraph{Case $\alpha < 1$}
This us dual: $x \mapsto x^\alpha$ is concave, the inequality of (\ref{eq:first_ineq_up}) is reversed, $x \mapsto x^{\frac{1}{\alpha}}$ is convex, and the inequality in (\ref{eq:jensen_proof_up}) is reversed, too. 
However, term $\frac{\alpha}{1-\alpha}$ is now positive. Thus, a dual argument shows that~(\ref{eq:conclu_proof_up}) holds.
\end{proof}

We can also extend the result of Theorem~\ref{thm:upperBound} to the limiting cases, i.e. to when $\alpha$ equals $1$ or $\infty$. 
\begin{corollary}
\label{cor:upperbound}
\begin{enumerate}
\item
We have:
\(\forall \pP \in \Omega(\DP), 
\forall \xA \in \DA \colon
\awae{\infty, g}{f', \pP}(\xA) \leq 
\awae{\infty, g}{f}(\xA) + \HH_{\infty}(\pP).\)

\item Moreover, if $g$ is unitary, then we have:
\(\forall \pP \in \Omega(\DP), 
\forall \xA \in \DA \colon
\awae{1, g}{f', \pP}(\xA) \leq 
\awae{1, g}{f}(\xA) + \HH_{1}(\pP).\)
\end{enumerate}
\end{corollary}

\begin{proof}
Lemma \ref{lemma:hlim} ensures the convergence of $\awae{\alpha, g}{f', \pP}(\xA)$ and of $\awae{\alpha,g}{f}(\xA)$ towards $\awae{\infty, g}{f', \pP}(\xA)$ and of $\awae{\infty, g}{f}(\xA)$ respectively when $\alpha$ tends to $\infty$. 
Moreover, it is known that the R\'enyi entropy $\HH_{\alpha}(\pP)$ of order $\alpha$ converges to the min-entropy $\HH_{\infty}(\pP)$ as $\alpha$ tends towards $\infty$. 
Thus, the result stated in item 1) follows from Theorem \ref{thm:upperBound} by letting $\alpha$ tend towards $\infty$. 
A similar argument concludes the proof for $\alpha$ tending towards $1$, in the case that $g$ is unitary. 
\end{proof}

Although one could have proved Corollaries \ref{cor:lowerbound} and \ref{cor:upperbound} with bespoke and somewhat different arguments, it is pleasing to see that our generalised conditional entropy makes such arguments uniform and reasonably simple.
Given a close approximation $f'$ of a function $f$, Theorems \ref{thm:lowerbound} and \ref{thm:upperBound} and Corollaries \ref{cor:lowerbound} and \ref{cor:upperbound} give us a lower bound and an upper bound for the entropy gain that a given virtual distribution provides. 
We can formalise this through a gain function $\Gamma_{\alpha,g}$, which indicates how much entropy we gain by adding a virtual input to the SMC computation --- as a function of the chosen probability distribution of this virtual input:
\begin{definition}
Let $f'$ be a close approximation of a function $f$.
Let $\alpha$ be in $\Rpe \cup \{ \infty \}$. 
Let us further assume that either $\alpha$ is different from $1$ or $g$ is unitary. Then, we define the function $\Gamma_{\alpha,g}$ for all $\pP$ in $\Omega(\DP)$ and $\xA$ in $\DA$ by:
\begin{equation}
\Gamma_{\alpha,g}(\pP, \xA) \coloneqq \awae{\alpha, g}{f', \pP}(\xA) - \awae{\alpha, g}{f}(\xA)
\end{equation}
\end{definition}

Then, under the assumptions of Theorem 2, we can summarise our above results as follows:
\begin{corollary}
Let $f'$ be a close approximation of $f$. 
Let $\alpha$ be in $\Rpe \cup \{ \infty \}$. 
Let us further assume that either $\alpha$ is different from $1$ or $g$ is unitary. 
Then, we have:
\begin{equation}
\forall \pP \in \Omega(\DP), 
\forall \xA \in \Omega(\DA) \colon
0 \leq \Gamma_{\alpha,g}(\pP, \xA) \leq \HH_\alpha(\pP)
\end{equation}
\end{corollary}

\begin{proof}
This result is a direct consequence of Theorems \ref{thm:lowerbound} and \ref{thm:upperBound} and Corollaries \ref{cor:lowerbound} and \ref{cor:upperbound}. 
\end{proof}

Let us now illustrate Theorems \ref{thm:lowerbound} and \ref{thm:upperBound} by means of a worked example. 
\begin{example}
\label{ex:minBounds}
Let us re-consider the scenario in Example~\ref{example:basic}
with the additive approximation $\tid{f}$ of $f$; in particular, $f'$ is a close approximation of $f$. 
We study the behaviour of the conditional min-entropy of the targeted inputs when we approximate $f$ with $\tid{f}$. In other words, we study the function $\awae{\infty, \id}{\tid{f}, \pP}$ for different distributions $\pP$. 
Since $\tid{f}$ is a close approximation of $f$, Theorems \ref{thm:lowerbound} and \ref{thm:upperBound} apply, and thus for all $\pP$ in $\Omega(\Z)$ and for all $\xA$ in $\DA$, we have:
\begin{equation}
\label{eq:exBoundsMin}
\awae{\infty, \id}{f}(\xA) \leq \awae{\infty, \id}{\tid{f}, \pP}(\xA) \leq \awae{\infty, \id}{f}(\xA) + \HH_{\infty}(\pP)
\end{equation}

In order to illustrate this property, we choose different distributions for $\phi$ that all have equal min-entropy:
\begin{eqnarray*}
\pP_1 &=& \{-2\colon \sfrac{1}{4}, 0\colon\sfrac{1}{4}, 2\colon\sfrac{1}{4}, 4\colon\sfrac{1}{4}\} \\
\pP_2 &=& \{-1\colon \sfrac{1}{4}, 0\colon\sfrac{1}{4}, 1\colon\sfrac{1}{4}, 2\colon\sfrac{1}{4}\} \\
\pP_3 &=& \{-3\colon \sfrac{1}{8}, -2\colon \sfrac{1}{8}, -1\colon\sfrac{1}{8}, 0\colon \sfrac{1}{4}, 1\colon\sfrac{1}{8}, 2\colon \sfrac{1}{8}, 3\colon\sfrac{1}{8}\}
\end{eqnarray*}

All those distributions have a min-entropy of $-\log(\sfrac{1}{4}) = 2$.
In Figure~\ref{fig:minBounds}, we plot the functions $\awae{\infty, \id}{f}$, $\awae{\infty, \id}{f} + 2$, and $\awae{\infty, \id}{\tid{f}, \pP_i}$ for all $1\leq i\leq 3$. We can verify that Equation~(\ref{eq:exBoundsMin}) indeed holds: for all $1\leq i\leq 3$, the function $\awae{\infty, \id}{\tid{f}, \pP_i}$ is contained between the functions $\awae{\infty, \id}{f}$ and $\awae{\infty,\id}{f} + \HH_{\infty}(\pP_i)$. 

Finally, note that although the three virtual distributions $\pP_i$ have equal min-entropy, they produce different values for $\awae{\infty,\id}{\tid{f}, \pP_i}$. From the plots we can clearly see, e.g., that $\pP_1$ produces higher entropy values than $\pP_2$ in general. 
This observation motivates us to seek optimal virtual distributions, which we focus on in the next section. 
\begin{figure}
\centering
\begin{tikzpicture}[scale=.9]
	\begin{axis}[
	scale=1,
	 ymin=0,
	 xmin=0,
	 xmax=30,
	  xlabel=$\xA$,
	  ylabel=entropy,
	  legend pos=outer north east]
	\addplot[blue] table [y=fx, x=x, mark=none]{minBounds.dat};
	\addlegendentry{$\awae{\infty, \id}{f}(\xA)$}
	\addplot[green!60!black, densely dotted, very thick] table [y=p1, x=x, mark=none]{minBounds.dat};
	\addlegendentry{$\awae{\infty,\id}{\tid{f}, \pP_1}(\xA)$}
	\addplot[orange, loosely dotted, ultra thick] table [y=p2, x=x, mark=none]{minBounds.dat};
	\addlegendentry{$\awae{\infty,\id}{\tid{f}, \pP_2}(\xA)$}
	\addplot[violet, dashed, thick] table [y=p3, x=x, mark=none]{minBounds.dat};
	\addlegendentry{$\awae{\infty,\id}{\tid{f}, \pP_3}(\xA)$}
	\addplot[red] table [y=fxp, x=x, mark=none]{minBounds.dat};
	\addlegendentry{$\awae{\infty, \id}{f}(\xA) + 2$}
\end{axis}
\end{tikzpicture}
\caption{Illustration of the bounds of $\awae{\infty, \id}{\protect\tid{f}, \pP_i}$ in the computation of $f(x, y, z) = x(2y + z) + 2z$ with $\pT$ and $\pS$ uniform over $\llrr{1}{30}$. For all $i$ in $\llbracket 1, 3 \rrbracket$, the function $\awae{\infty, \id}{\protect\tid{f}, \pP_i}$ is contained between $\awae{\infty, \id}{f}$ and $\awae{\infty,\id}{f} + \HH_{\infty}(\pP_i)$. 
}
\label{fig:minBounds}
\end{figure}
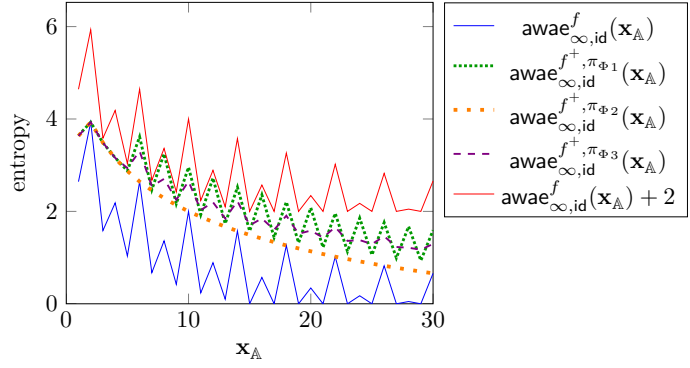
\end{example}

\section{Optimal Trade-Off Between Accuracy and Privacy}
\label{section:optimal}

So far, we developed a means of replacing a function $f$ by an approximating function $f'$ which resorts to additional, virtual inputs governed by some distribution. 
We showed that such approximations enable us to protect the privacy of the targeted inputs. 
These benefits are hampered by the fact that the approach introduces a distortion on the output for function $f$ when computing with $f'$ instead.
The participants of the SMC computation from set $\bP$ are either eager to learn the actual output of function $f$ or they would tolerate only a certain difference between the outputs of $f$ and $f'$, and these demands would typically be informed by the use-context of the SMC computation. 

Therefore, we need to have methods by which we can control the support and the distribution of the virtual input, in order to measure and control both the distortion that $f'$ and $\pi_{\Phi}$ introduce, and the privacy gain that it offers over using $f$ for SMC instead. 
We therefore develop now the formalism needed for studying the inherent trade-off between the accuracy of the output and the privacy of supplied inputs. 
We also recall that Assumption \ref{ass:beliefs} ensures that any of the parties can perform the methods we introduce next and compute optimal virtual distributions. 
When replacing the computation of a function $f$ by an approximation $f'$, the output accuracy is directly influenced by the choice of $f'$. 
A function $f'$ that is the constant $0$ function, e.g., would not reveal anything about the inputs, but be very inaccurate.
\begin{assumption}
In the remaining paper, we will focus on additive approximations $\tid{f}$ of $f$.
\end{assumption}

\noindent This is a natural assumption to make, it simplifies our problem, as shown in (\ref{eq:distortion_additive}), and enables us to characterise optimal virtual distributions. 
We will also propose some practical methods from optimisation for discovering virtual distributions that realise this trade-off in an optimal manner.

\subsection{Maximal and Optimal Distortion}

We want to contain 
the distortion introduced by computing $\tid{f}$ instead of $f$. 
Formally, for a given virtual input $\phi$ with distribution $\pP$, we will tolerate a certain distortion threshold $\Delta$ in $\Np$ that serves as upper bound for the maximal absolute difference $\xi(f, \tid{f})$ between the output of $f$ and $\tid{f}$, i.e.\ $\xi(f, \tid{f}) \leq \Delta$ where:
\begin{equation}
\label{equ:maxdiff}
\xi(f, \tid{f}) \coloneqq
\max_{\substack{
	\phi \in \supp(\pP) \\ 
	\xx \in D }} 
	| f(\xx) - \tid{f}(\xx,\phi) |
\end{equation}

\noindent and where $\supp(\pP) \coloneqq \{ \phi \in \DP \mid \pP(\phi) > 0 \}$ denotes the support of $\pP$. 
For the additive approximation $\tid{f}$ of $f$, we can see that $\xi(f, \tid{f})$ equals $\max \{|\phi| \mid \phi \in \supp(\pP)\}$. Thus, we have:
\begin{equation}
\label{eq:distortion_additive}
\xi(f, \tid{f}) \leq \Delta 
\iff
\supp(\pP) \subseteq \llrr{-\Delta}{+\Delta}
\end{equation}

For such additive approximation $\tid{f}$, our examples suggested that different distributions $\pP$ for a virtual input $\phi$ can yield different privacy gains for the targeted inputs. 
We are thus interested in studying the influence of the distribution $\pP$ of the virtual input on the behaviour of the leakage measure $\awae{\alpha, g}{\tid{f}, \pP}$. To that end, we first  want to evaluate how much privacy is being protected by $\tid{f}$ and $\pi_{\Phi}$ within a distortion threshold $\Delta$. We can do this through a metric, our objective function for optimisation, that uses the weighted average of $\awae{\alpha, g}{\tid{f}, \pP}$ over all the values of $\xA$. 
\begin{definition}
Let $\alpha$ be in $\Rpe \cup \{\infty\}$. 
Function $\obj_{\alpha,g} \colon \Omega(\DP) \to \Rp$ is defined, for all $\pP$ in $\Omega(\DP)$, as:
\begin{equation}
\label{eq:obj_f}
\obj_{\alpha,g}(\pP) = 
\sum_{\xA \in \DA} p(\xA) \cdot \awae{\alpha, g}{\tid{f}, \pP}(\xA)
\end{equation}
\end{definition}

The targeted parties in $\bT$ --- and perhaps others --- now want to find a distribution $\pP$ that will be optimal for this given metric, under the constraint that the distortion should remain below the threshold $\Delta$. 

Entropies, as mathematical functions, are such that the larger their output is, the less do we actually know.
Therefore, we mean to find a \emph{global maximum} of the above metric, subject to the distortion-bound constraint. This ensures that an attacker has, on average, the least information gain for this from all possible virtual distributions.
Using the equivalence in Equation (\ref{eq:distortion_additive}), this naturally leads us to the following optimisation problem. 
\begin{definition}
\label{def:opti}
Let $\Delta$ be in $\Np$, 
let $\alpha$ be in $\Rpe \cup \{\infty\}$. 
Then we denote by $\OP_{\alpha,g}(\Delta)$ the optimisation problem:
\begin{equation}
\label{equ:opti}
\begin{aligned}
& \underset{\pP \in \Omega(\llrr{-\Delta}{+\Delta})}{\textnormal{maximise}}
& & \obj_{\alpha, g}(\pP)
\end{aligned}
\end{equation}

\item We write $\omega_{\alpha,g}$ for the optimal objective value in~(\ref{equ:opti}). 
\end{definition}

Note that this optimisation problem can equivalently be expressed as optimising the $2\Delta + 1$ values of distribution $\pP$:
\begin{equation}
\begin{aligned}
& \underset{(\pP(i))_{-\Delta \leq i \leq +\Delta}}{\textnormal{maximise}}
& & \obj_{\alpha, g}(\pP)\\
& \text{subject to}
& & \sum_{i \in \llrr{-\Delta}{+\Delta}} \pP(i) = 1\\
& \text{and}
& & \forall i \in \llrr{-\Delta}{+\Delta} \colon 0 \leq \pP(i) \leq 1
\end{aligned}
\end{equation}

	\subsection{Computing Optimal Virtual Distributions}
We now discuss methods for solving this optimisation problem and computing optimal virtual distributions, where we distinguish between the cases in which $\alpha$ is $\infty$ or greater or equal to $1$. 

\ourparagraph{Optimal Virtual Input Randomisation when $1 \leq \alpha < \infty$} 
For a gain function $g\colon \cW \times \DT \to [0,1]$, let us study the objective function of $\OP_{\alpha,g}(\Delta)$. 
We recall that for all $\xA$ in $\DA$ and for $\V_{\alpha,g}$ as defined in~(\ref{equ:Vag}), we have:
\[
\awae{\alpha, g}{\tid{f}, \pP} (\xA)
= \frac{\alpha}{1-\alpha} \cdot \log\bigl ( \V_{\alpha,g}(\XT \mid O', \xA)\bigr )
\]

\noindent and where, for all $\xT$ in $\DT$, the term $p(o' \mid \xT, \xA)$ is a linear function of $\pP$, namely:
\[
p(o' \mid \xT, \xA) = 
\sum_{\substack{\xS, \phi \\ \tid{f}(\xA, \xT, \xS, \phi) = o'}} 
p(\xS) \cdot p(\phi)
\]

\noindent Below, we may write $p(o', \mid \xT, \xA)[\pP]$ for $p(o', \mid \xT, \xA)$ in order to make this linear dependency on $\pP$ explicit. 

We thus have a non-linear and non-convex optimisation problem with linear constraints and where the objective function is twice continuously differentiable almost everywhere. 
Sequential Quadratic Programming (SQP) \cite{nocedal2006sequential,boggs1995sequential} would thus seem like an adequate and simple solution for finding a local optimum for our optimisation problem. 
However, SQP requires the constraints and the objective function to be twice continuously differentiable, which is not the case of our objective function: for all $\alpha > 1$ and all integer $n > 1$, the function $y \mapsto \lVert y \rVert_\alpha$ is not differentiable at the origin even when restricted to $(\Rp)^n \to \Rp$. 
Consequently, our objective function is not differentiable at the points $\pP_0$ in $\Omega(\DP)$ such that $\pP_0$ makes $p(o', \mid \xT, \xA)$ be $0$ but where $p(o', \mid \xT, \xA)$ is not always $0$, i.e., when:
\[ 
\left( p(o', \mid \xT, \xA)[\pP_0] = 0 \right) \wedge
\left( \exists \pP_1 \in \Omega(\DP) \colon 
p(o', \mid \xT, \xA)[\pP_1] > 0 \right)
\]

\noindent We will address this by smoothening the objective function in~(\ref{eq:obj_f}) through a non-zero offset vector $\bm{\delta}$ in $(\Rp)^{|\DT|}$ that is added to the argument of the $\alpha$-norm --- the expression in~(\ref{equ:Vag}) with $O'$ instead of $O$. This approximation is then twice continuously differentiable everywhere. 
We introduce some definitions for formalising this:
\begin{definition}
Let $\alpha$ be in $]1, \infty[$. 
Let $\bm{\delta}\not=0$ be in $(\Rp)^{|\DT|}$.
\begin{enumerate}
\item
Let $\pP$ be in $\Omega(\DP)$. 
For all $\xA$ in $\DA$, we define:
\begin{eqnarray}
\V_{\alpha, g}^{\bm{\delta}}(\XT \mid O', \xA)
&\coloneqq&
\sum_{o'} 
\norm{ 
	\bm{\delta} + 
	\vv{ \sum_\xT p(\xT) \cdot p(o' \mid \xT, \xA) \cdot g(w, \xT) }_w 
}_\alpha\\
\awae{\alpha, g}{\tid{f}, \pP, \bm{\delta}} (\xA)
&\coloneqq& \frac{\alpha}{1-\alpha} \cdot \log \bigl (\V_{\alpha, g}^{\bm{\delta}}(\XT \mid O', \xA)\bigr )
\end{eqnarray}

\item We define the function $\obj_{\alpha, g}^{\bm{\delta}} \colon \Omega(\DP) \to \Rp$ for all $\pP$ in $\Omega(\DP)$ as:
\begin{equation}
\obj_{\alpha, g}^{\bm{\delta}}(\pP) \coloneqq
\sum_{\xA \in \DA} p(\xA) \cdot \awae{\alpha, g}{\tid{f}, \pP, \bm{\delta}}(\xA)
\end{equation}

\noindent For $\Delta$ in $\Np$, we define $\OP_{\alpha, g}^{\bm{\delta}}(\Delta)$ as the following optimisation problem:
\begin{equation}
\label{equ:opti_delta}
\begin{aligned}
& \underset{\pP \in \Omega(\llrr{-\Delta}{+\Delta})}{\textnormal{maximise}}
& & \obj_{\alpha, g}^{\bm{\delta}}(\pP)
\end{aligned}
\end{equation}

\noindent We write $\omega_{\alpha, g}^{\bm{\delta}}$ for the global maximum of the optimisation problem $\OP_{\alpha, g}^{\bm{\delta}}(\Delta)$. 
\end{enumerate}

\end{definition}

Using the above optimisation problems, we are now able to approximate the result of the original problem in (\ref{equ:opti}) with an arbitrary accuracy by choosing the value of $\bm{\delta}$. We formalise this next:
\begin{theorem}
\label{thm:opt_delta}
Let $\alpha$ be in $]1, \infty[$. 
Let $g$ be a $\beta$-positive gain function (as defined in Definition \ref{def:g_functions}).
Let $\Delta$ be in $\Np$ and let $\bm{\delta}$ be the vector in $(\Rpe)^{|\DT|}$ whose $|\DT|$ components all equal $\delta$ in $\Rpe$.  Then, we have:
\begin{equation}
\forall \epsilon > 0 \colon
\left(
\delta \leq (1 - \frac{1}{\alpha}) \cdot \frac{\epsilon \cdot \beta \cdot \ln (2)}{|\DOp| \cdot |\cW|}
\right)
\implies
\left(
| \omega_{\alpha,g} - \omega_{\alpha,g}^{\bm{\delta}} | \leq \epsilon
\right)
\end{equation}

\noindent where $\ln (2)$ refers to the natural logarithm of $2$. 
\end{theorem}

\begin{proof}
Let $\pP$ be in $\Omega(\DP)$, let $o'$ be in $\DOp$, and let $\xA$ be in $\DA$
For sake of convenience, let us define the vector:
\[ W_{\xA}^{o'} \coloneqq \vv{ \sum_\xT p(\xT) \cdot p(o' \mid \xT, \xA) \cdot g(w, \xT) }_w  \]

\noindent First, as all the components of the vectors are non-negative, we have:
\[
\norm{W_{\xA}^{o'} + \bm{\delta}}_\alpha \geq
\norm{W_{\xA}^{o'}}_\alpha \]

\noindent Since $\alpha$ is greater than $1$, we know that $\frac{\alpha}{1-\alpha}$ is negative, and thus:
\[ \awae{\alpha, g}{\tid{f}, \pP, \bm{\delta}} (\xA) \leq 
\awae{\alpha, g}{\tid{f}, \pP} (\xA) \]

\noindent Moreover, application of the triangular inequality yields:
\[ 
\sum_{o'} \norm{W_{\xA}^{o'} + \bm{\delta}}_\alpha \leq 
\sum_{o'} \norm{W_{\xA}^{o'}}_\alpha + \sum_{o'} \norm{\bm{\delta}}_\alpha
\]

\noindent Applying logarithm and multiplying by $\frac{\alpha}{1-\alpha}$ on both sides, we obtain:
\[ 
\awae{\alpha, g}{\tid{f}, \pP, \bm{\delta}} (\xA)
\geq 
\frac{\alpha}{1-\alpha} \cdot \log \left(
\sum_{o'} \norm{W_{\xA}^{o'}}_\alpha + \sum_{o'} \norm{\bm{\delta}}_\alpha \right)
\]

\noindent However, for all $a$ and $b$ in $\Rpe$, we have $\log(a+b) = \log(a)+\log(1+\frac{b}{a})$. Therefore, we conclude that:
\[ 
\awae{\alpha, g}{\tid{f}, \pP, \bm{\delta}} (\xA)
\geq 
\awae{\alpha,g}{\tid{f}, \pP} (\xA)
+
\frac{\alpha}{1-\alpha} \cdot \log \left( 1+
\frac{\sum_{o'} \norm{\bm{\delta}}_\alpha}{\sum_{o'} \norm{W_{\xA}^{o'}}_\alpha}
\right)
\]

\noindent Rearranging the terms and summing over $\xA$ gives us:
\[
\obj_{\alpha,g}(\pP) -
\obj_{\alpha,g}^{\bm{\delta}}(\pP)
\leq
\sum_{\xA} p(\xA) \cdot
\frac{\alpha}{\alpha-1} \cdot \log \left( 1+
\frac{\sum_{o'} \norm{\bm{\delta}}_\alpha}{\sum_{o'} \norm{W_{\xA}^{o'}}_\alpha}
\right)
\]

\noindent Moreover, for all $x$ in $\Rpe$, we know that $\log(1+x) \leq x/\ln (2)$. Thus, we infer:
\begin{equation}
\label{eq:proof_ineq}
\obj_{\alpha,g}(\pP) -
\obj_{\alpha,g}^{\bm{\delta}}(\pP)
\leq
\sum_{\xA} p(\xA) \cdot 
\frac{\alpha}{\alpha-1} \cdot
\frac{\sum_{o'} \norm{\bm{\delta}}_\alpha}{\ln (2) \cdot \sum_{o'} \norm{W_{\xA}^{o'}}_\alpha}
\end{equation}

Furthermore, for all $p$ in $[1, \infty]$ and $n$ in $\Np$, we get from the topological equivalence of the norms in finite dimension that for all $x$ in $\R^n$, we have $\norm{x}_p \geq \norm{x}_1 \cdot n^{\frac{1}{p} - 1}$. Therefore:
\[
\norm{W_{\xA}^{o'}}_\alpha \geq \norm{W_{\xA}^{o'}}_1 \cdot |\cW|^{\frac{1}{\alpha}-1}
\]

\noindent Now, we know that:
\begin{eqnarray*}
\norm{W_{\xA}^{o'}}_1 &=& \sum_w \sum_\xT p(\xT) \cdot p(o' \mid \xT, \xA) \cdot g(w, \xT) \\
&=& \sum_\xT p(\xT) \cdot p(o' \mid \xT, \xA) \cdot \left( \sum_w g(w, \xT) \right)
\end{eqnarray*} 

\noindent Since $g$ is $\beta$-positive, we obtain:
\[
\sum_{o'} \norm{W_{\xA}^{o'}}_\alpha \geq \beta \cdot |\cW|^{\frac{1}{\alpha}-1}
\]

On the other hand, by definition of $\bm{\delta}$ we have:
\[ \sum_{o'} \norm{\bm{\delta}}_\alpha = \delta \cdot |\DOp|\cdot |\cW|^{\frac{1}{\alpha}} \]

\noindent and thus Equation (\ref{eq:proof_ineq}) becomes:
\begin{eqnarray}
\obj_{\alpha,g}(\pP) -
\obj_{\alpha,g}^{\bm{\delta}}(\pP)
&\leq&
\sum_{\xA} p(\xA) \cdot 
\frac{\alpha}{\alpha-1} \cdot
\frac{\delta \cdot |\DOp|\cdot |\cW|^{\frac{1}{\alpha}}}{\ln (2) \cdot \beta \cdot |\cW|^{\frac{1}{\alpha}-1}} \nonumber \\
\label{eq:proof_end}
&\leq& 
\frac{\alpha}{\alpha-1} \cdot
\frac{\delta \cdot |\DOp|\cdot |\cW|}{\ln 2 \cdot \beta}
\end{eqnarray}

Consider now any $\epsilon$ in $\Rpe$. 
In order for $\obj_{[\alpha, g]}(\pP) -
\obj_{\alpha, g}^{\bm{\delta}}(\pP)$ 
not to exceed $\epsilon$, Equation (\ref{eq:proof_end}) ensures that it suffices to have:
\begin{equation}
\label{eq:conditionDelta}
\delta \leq (1 - \frac{1}{\alpha}) \cdot \frac{\epsilon \cdot \beta \cdot \ln (2)}{|\DOp| \cdot |\cW|} 
\end{equation}

Finally, the reverse triangle inequality applied on functions $\obj_{\alpha, g}$ and $\obj_{\alpha, g}^{\bm{\delta}}$ with the uniform norm yields:
\[ | \omega_{\alpha,g} - \omega_{\alpha,g}^{\bm{\delta}} | \leq 
\max_{\pP} | \obj_{\alpha, g}(\pP) - \obj_{\alpha, g}^{\bm{\delta}}(\pP) |
 \]
and thus the condition in (\ref{eq:conditionDelta}) implies $| \omega_{\alpha,g} - \omega_{\alpha,g}^{\bm{\delta}} | < \epsilon$. 

\end{proof}

The last theorem states that, if we are able to solve the optimisation problem $\OP_{\alpha,g}^{\bm{\delta}}(\Delta)$ for any non-zero offset vector $\bm{\delta}$ in $\Rp^{|\DT|}$, then we are able to approximate the optimal outcome of the original optimisation problem $\OP_{\alpha,g}(\Delta)$ with arbitrary precision. 
We now present a method for solving the approximate optimisation problems $\OP_{\alpha,g}^{\bm{\delta}}(\Delta)$. 
\begin{method}
\label{met:delta}
Let us consider the optimisation problem $\OP_{\alpha,g}^{\bm{\delta}}(\Delta)$ of (\ref{equ:opti_delta}) where $\alpha$ is in $]1, \infty[$
The objective function $\obj_{\alpha,g}^{\bm{\delta}}$ is twice differentiable and the constraints are linear. 
Thus, we may apply SQP \cite{nocedal2006sequential,boggs1995sequential} to find a local optimum for $\OP_{\alpha,g}^{\bm{\delta}}(\Delta)$. 
However, as the objective function $\obj_{\alpha, g}^{\bm{\delta}}$ is non-convex, we will use a globalisation technique known as the basin-hopping algorithm \cite{wales1997global}. In order to respect the linear constraints of this problem, the starting points of this algorithm will be drawn from a symmetric Dirichlet distribution. 

\end{method}

This computational method lets us
solve optimisation problems of the form $\OP_{\alpha,g}^{\bm{\delta}}(\Delta)$. 
Consequently, Theorem \ref{thm:opt_delta} enables us to build a method for solving our original optimisation problem $\OP_{\alpha,g}(\Delta)$ with an arbitrary precision $\epsilon$, which we formalise in the next method:
\begin{method}
\label{met:alpha}
We seek a solution of the optimisation problem $\OP_{\alpha,g}(\Delta)$ where $\alpha$ is in $[1, \infty[$ and $g$ is a $\beta$-positive gain function. 

\ourparagraph{Case $\alpha > 1$}
We will solve $\OP_{\alpha,g}(\Delta)$ with a given accuracy $\epsilon > 0$. 
In other words, a solution $\pP$ should satisfy $|\obj_{\alpha,g}(\pP) - \omega_{\alpha,g}| \leq \epsilon$. 
First, let us choose $\delta$ in $\Rpe$ such that:
\[\delta \leq (1 - \frac{1}{\alpha}) \cdot \frac{\epsilon \cdot \beta}{|\DOp| \cdot |\cW|}\]

\noindent Let $\bm{\delta}$ be the vector in $\R^{|\DT|}$ whose components all equal $\delta$. 
We apply Method~\ref{met:delta} in order to solve the optimisation problem $\OP_{\alpha,g}^{\bm{\delta}}(\Delta)$ . 
Let $\pP$ be the solution output by Method \ref{met:delta}. 
By virtue of Theorem \ref{thm:opt_delta}, we have $|\obj_{\alpha,g}(\pP) - \omega_{\alpha,g}| \leq \epsilon$. 

\ourparagraph{Case $\alpha = 1$}
Let $g$ be unitary. We can solve $\OP_{1,g}(\Delta)$ using the same procedure as that of Method \ref{met:delta} since the objective function $\obj_{1, g}$ and the constraints of the problem are twice continuously differentiable. 
\end{method}

Now that we are able to solve the optimisation problem $\OP_{\alpha,g}(\Delta)$ when $\alpha \geq 1$ is finite, we now turn our attention to the case of $\alpha=\infty$. 
In the same way as Method \ref{met:alpha} builds on Method \ref{met:delta} to approximate a solution, our next idea will be to approximate the optimal result of $\OP_{\infty,g}(\Delta)$ with a multiple of that of $\OP_{\alpha,g}(\Delta)$ for a sufficiently large $\alpha$. 

First, we introduce $\overline{\OP}_{\alpha,g}(\Delta)$, a slightly modified version of $\OP_{\alpha,g}(\Delta)$ whose objective function is a multiple of $\obj_{\alpha,g}$. 
Then, we prove that the solutions of $\overline{\OP}_{\alpha,g}(\Delta)$ converge towards a solution of $\OP_{\infty,g}(\Delta)$. Moreover, we make the convergence rate explicit for computational purposes. 
We define $\overline{\OP}_{\alpha, g}(\Delta)$ next:
\begin{definition}
Let $\alpha$ be in $]1, \infty[$. 
\begin{enumerate}
\item We define the function 
$\overline{\obj}_{\alpha,g} \colon \Omega(\DP) \to \Rp$ for all $\pP$ in $\Omega(\DP)$ as:
\begin{equation}
\overline{\obj}_{\alpha,g}(\pP) \coloneqq
\frac{\alpha - 1}{\alpha} \cdot
\obj_{\alpha,g}(\pP)
\end{equation}

\item For $\Delta$ in $\Np$, the optimisation problem $\overline{\OP}_{\alpha,g}(\Delta)$ is:
\begin{equation}
\begin{aligned}
& \underset{\pP \in \Omega(\llrr{-\Delta}{+\Delta})}{\textnormal{maximise}}
& & \overline{\obj}_{\alpha,g}(\pP)
\end{aligned}
\end{equation}

\noindent We write $\overline{\omega}_{\alpha,g}$ denote the optimal objective value for $\overline{\OP}_{\alpha,g}(\Delta)$. 
\end{enumerate}
\end{definition}

From this definition it is clear that $\OP_{\alpha, g}(\Delta)$ and $\overline{\OP}_{\alpha, g}(\Delta)$ are equivalent optimisation problems, in that:
\begin{equation}
\label{eq:link_op_prime}
\overline{\omega}_{\alpha,g} = \frac{\alpha - 1}{\alpha} \cdot \omega_{\alpha,g}
\end{equation}

In fact, the optimal values of $\overline{\OP}_{\alpha,g}(\Delta)$ under-approximate that of $\OP_{\infty,g}(\Delta)$, with an error rate dominated by $\frac{1}{\alpha}$:
\begin{theorem}
\label{thm:cvg}
Let the functions $\tau, \theta \colon ] 1, +\infty [ \to \Rp$ be defined as $\tau(\alpha) = \overline{\omega}_{\alpha,g}$ and $\theta(\alpha) = \abs{\omega_{\infty,g} - \overline{\omega}_{\alpha,g}}$. 
Then, for all $\alpha > 1$, we have $\tau (\alpha)\leq \omega_{\infty,g}$, $\lim_{\alpha\to \infty} \tau(\alpha) = \omega_{\infty,g}$, and $\theta(\alpha) = \mathcal{O}(\frac{1}{\alpha})$. 

\end{theorem}

\begin{proof}
Let $\pP$ be in $\Omega(\Z)$. 
To simplify notation,
we define the vector $Y_{\xA}^{o'}$ for all $\xA$ in $\DA$ and $o'$ in $\DOp$ as:
\[ Y_{\xA}^{o'} \coloneqq \vv{ \sum_\xT p(\xT \mid o', \xA) \cdot g(w, \xT) }_w  \]

For all $\alpha$ in $]1, \infty]$ and $\pP$ in $\DP$, we have by definition that $\overline{\obj}_{\alpha, g}(\pP)$ equals:
\[ \overline{\obj}_{\alpha,g}(\pP) = - \sum_{\xA} p(\xA) \cdot  
 \log \left( 
\sum_{o'} p(o' \mid \xA) \cdot \norm{Y_{\xA}^{o'}}_\alpha
 \right)
\]

We know that in finite dimension, all the norms are topologically equivalent. In particular, for all $n$ in $\Np$, $x$ in $\R^n$, and $p$ in $]1, \infty[$, we have:
\[
\norm{x}_\infty \leq \norm{x}_p \leq \norm{x}_\infty \cdot n^{\frac{1}{p}}
\]

\noindent Let $\alpha$ be in $]1, \infty[$. 
We thus have:
\[
\norm{Y_{\xA}^{o'}}_\infty
\leq 
\norm{Y_{\xA}^{o'}}_\alpha
\leq
\norm{Y_{\xA}^{o'}}_\infty \cdot |\cW|^{\frac{1}{\alpha}}
\]

\noindent and thus:
\[
\obj_{\infty,g}(\pP) - \frac{1}{\alpha} \cdot \log(|\cW|)
\leq
\overline{\obj}_{\alpha,g}(\pP)
\leq 
\obj_{\infty,g}(\pP)
\]

\noindent From this inequality, we can see that $\tau(\alpha) \leq \omega_{\infty,g}$ for all $\alpha > 1$.
Moreover:
\begin{eqnarray*}
\theta(\alpha) &\leq&
\obj_{\infty,g}(\pP) - \left(\obj_{\infty,g}(\pP) - \frac{1}{\alpha} \cdot \log(|\cW|)\right) \\
&\leq& \frac{1}{\alpha} \cdot \log(|\cW|)
\end{eqnarray*}
and thus $\theta(\alpha) = \mathcal{O}(\frac{1}{\alpha})$. 
In particular, $\theta$ converges to $0$ as $\alpha$ goes to infinity, which ensures that $\tau$ converges at infinity such that $\lim_{\tau\to \infty} \tau(\alpha) = \omega_{\infty,g}$. 
Moreover, for any $\epsilon$ in $\Rpe$, in order to have $\theta(\alpha) \leq \epsilon$, it suffices to have:
\[
\alpha \geq \frac{1}{\epsilon} \log(|\cW|)
\]
\end{proof}

From this theorem, we can build a method for solving the optimisation problems of the form $\OP_{\infty,g}(\Delta)$. Indeed, even though the objective function $\obj_{\infty,g}$ is not twice differentiable, we can approximate the solution of that optimisation problem with that of $\overline{\OP}_{\alpha,g}(\Delta)$ for a sufficiently large $\alpha$. 
We recall that, by~(\ref{eq:link_op_prime}), the optimal value of the latter problem is a multiple of that of $\OP_{\alpha,g}(\Delta)$, which we can solve with Method~\ref{met:alpha}. 
However, Method~\ref{met:alpha} also requires a non-zero accuracy threshold. Thus, for a given $\epsilon$ in $\Rpe$, we will invoke Method~\ref{met:alpha} in order to solve $\OP_{\alpha,g}(\Delta)$ with accuracy  $\frac{\alpha}{\alpha -1} \cdot \frac{\epsilon}{2}$, and we will take advantage of Theorem~\ref{thm:cvg} to ensure that the output of our method indeed approximates the optimal objective value with precision $\epsilon$. 
We formalise this idea next: 
\begin{method}
\label{met:infinity}
Let $\epsilon$ be in $\Rpe$ and let us assume that we wish to solve $\OP_{\infty,g}(\Delta)$ with accuracy $\epsilon$, i.e.\ that the solution $\pP$ we get satisfies $|\obj_{\infty,g}(\pP) - \omega_{\infty,g} | \leq \epsilon$. 
First, we take some $\alpha > 1$ which satisfies:
\begin{equation}
\label{eq:alpha_greater}
\alpha \geq \frac{2}{\epsilon} \cdot \log(|\cW|)
\end{equation}

Then, we invoke Method~\ref{met:alpha} in order to solve $\OP_{\alpha,g}(\Delta)$ with accuracy $\frac{\alpha}{\alpha -1} \cdot \frac{\epsilon}{2}$. 
Let $\pP$ be an optimal solution for this produced by Method \ref{met:alpha}. 
Then, $\pP$ is an optimal solution of $\OP_{\infty,g}(\Delta)$ with accuracy $\epsilon$. 
\end{method}

\begin{proof}
As $\pP$ is the output of Method \ref{met:alpha}, we know that 
\( |\obj_{\alpha,g}(\pP) - \omega_{\alpha,g}| \leq \frac{\alpha}{\alpha -1} \cdot \frac{\epsilon}{2} \).
Multiplying both sides by $\frac{\alpha - 1}{\alpha}$ yields:
\[ |\overline{\obj}_{\alpha,g}(\pP) - \overline{\omega}_{\alpha,g}| \leq \frac{\epsilon}{2} \]

\noindent Moreover, by virtue of Theorem \ref{thm:cvg} and as we have Equation (\ref{eq:alpha_greater}), we know that 
\( | \overline{\omega}_{\alpha,g} - \omega_{\alpha,g} | \leq \frac{\epsilon}{2} \). 
Finally, we know that: 
\[ \overline{\obj}_{\alpha,g}(\pP) \leq \obj_{\infty,g}(\pP) \leq \omega_{\infty,g}\]

\noindent Appealing to the triangular inequality, we then get:
\begin{eqnarray*}
| \obj_{\infty,g}(\pP) - \omega_{\infty,g} | 
&\leq& | \overline{\obj}_{\alpha,g}(\pP) - \omega_{\infty,g} | \\
&\leq& | \overline{\obj}_{\alpha,g}(\pP) - \overline{\omega}_{\alpha,g} | + | \overline{\omega}_{\alpha,g} - \omega_{\infty,g} | \\
&\leq& \epsilon
\end{eqnarray*}
\end{proof}

The next example illustrates how the solution of $\OP_{\infty,g}(\Delta)$ is approximated by the successive solutions of $\overline{\OP}_{\alpha,g}(\Delta)$ for different values of $\alpha$. 
It is worth noting that the calculation of $\alpha$-norms involves the exponentiation of real numbers ranged in $[0, 1]$ which can quickly be rounded to $0$ for large values of $\alpha$. 
In order to mitigate against the effects of such numerical errors, results reports in this paper rely on use of the \texttt{mpmath} Python library, which enables us to perform arbitrary-precision floating-point arithmetic \cite{mpmath}. 
\begin{example}
\label{ex:minNonDiff}
Let us consider $3$ parties $X$, $Y$ and $Z$ with respective inputs $x$, $y$ and $z$, and where $\bA = \{X\}$ is attacking $\bT = \{Y\}$ under spectator $\bS = \{Z\}$. 
Let $\DA = \DT = \DS = \{1, 2\}$. 
Let $\pT$ and $\pS$ be linear distributions over their domains and let $\pA = \{1\colon 1\}$ be a point-mass distribution centred in $1$. 
Function $f\colon \Z^3 \to \Z$ is defined by $f(x, y, z) = 5xy - 2yz$. 

We study the influence of distributions $\pP$ for virtual inputs over $\Omega(\{0, 1\})$ on $\obj_{\infty, \id}$ produced by the output randomisation $\tid{f}$. 
Such two-dimensional distributions $\pP$ will be represented by a single real $r$ in $[0, 1]$, which fully characterises $\pP$ as $\{0\colon r, 1\colon 1-r\}$. 
We evenly discretise the interval $[0, 1]$ into $201$ values for $r$, and we plot the values of $\obj_{\infty, \id}$ in Figure \ref{fig:ex2Dp}.  
In order to see the influence of our smoothing method, we also plot the values of $\overline{\obj}_{\alpha, \id}$ for different values of $\alpha$. 
We can notice that, as suggested by our previous discussion and by Theorem \ref{thm:cvg}, the approximations $\overline{\obj}_{\alpha,\id}$ of $\obj_{\infty, \id}$ are functions that are twice differentiable and that also under-approximate $\obj_{\infty, \id}$. Moreover, larger values of $\alpha$ produce more accurate approximations of the original objective function. 
\begin{figure}
\centering
\begin{tikzpicture}[scale=.9]
	\begin{axis}[
	scale=1.3,
	 ymin=0,
	 xmin=0,
	 xmax=1,
	  xlabel=$\pP(0)$,
	  ylabel=entropy,
	  legend pos=outer north east]
	\addplot[blue] table [y=wme, x=x, mark=none]{op2Dp.dat};
	\addlegendentry{$\obj_{\infty, \id}(\pP)$}
	\addplot[violet, dashed, thick] table [y=p3, x=x, mark=none]{op2Dp.dat};
	\addlegendentry{$\overline{\obj}_{3, \id}(\pP)$}
	\addplot[green!60!black, loosely dotted, ultra thick] table [y=p4, x=x, mark=none]{op2Dp.dat};
	\addlegendentry{$\overline{\obj}_{4, \id}(\pP)$}
	\addplot[orange, densely dotted, very thick] table [y=p10, x=x, mark=none]{op2Dp.dat};
	\addlegendentry{$\overline{\obj}_{10, \id}(\pP)$}
\end{axis}
\end{tikzpicture}
\caption{Approximation of $\obj_{\infty, \id}$ by $\overline{\obj}_{p, \id}$ for $p$ in $\{3, 4, 10\}$ while computing $f(x, y, z) = 5xy - 2yz$ with linear distributions over $\{1, 2\}$ for $\pT$ and $\pS$ and $\pA=\{1\colon1\}$. }
\label{fig:ex2Dp}
\end{figure}
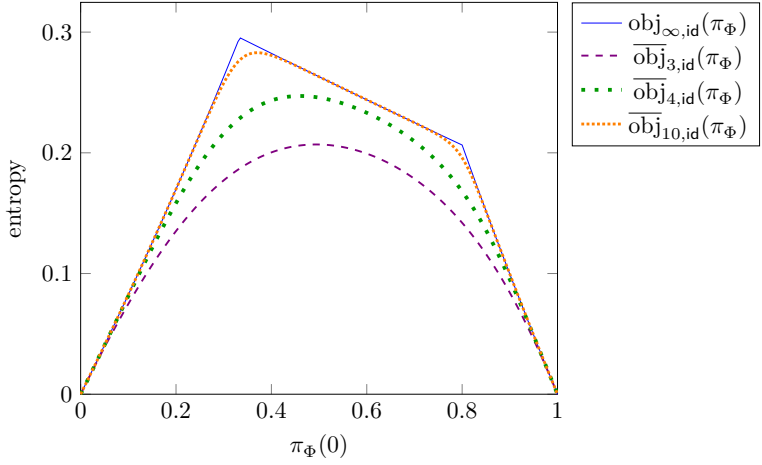

\end{example}

Let us now illustrate how the methods we developed here help us to find virtual distributions that protect the inputs' privacy optimally, given a some accuracy bound on the distorted output.
\begin{example}
\label{example:basic}
Let us consider $3$ parties $X$, $Y$ and $Z$ with respective inputs $x$, $y$ and $z$, and where $\bA = \{X\}$ is attacking $\bT = \{Y\}$ under spectator $\bS = \{Z\}$. 
Let $\DA = \DT = \DS = \llrr{1}{30}$.
Let $\pT$ and $\pS$ be linear distributions over their domains and for the sake of the example, let $\pA = \{5\colon 1\}$ be a point-mass distribution centred in $5$. 
Let us consider the function $f\colon \Z^3 \to \Z$ defined by $f(x, y, z) = x(3y - 5z) + 2z$. 
Let $\cW=\{0, 1\}$ be a set of allowable guesses and let $g\colon \cW \times \DT \to [0, 1]$ be the gain function defined for all $w$ in $\cW$ and $\xT$ in $\DT$ as:
\[
g(w, \xT) = \left\{\begin{array}{ll}
        1 & \text{if } w \equiv \xT \mod 2\\
        0 & \text{otherwise}
        \end{array}
        \right.
\]

In other words, this gain function $g$ measures the information that an attacker has on the least significant bit of the secret $\xT$.
More generally, we can consider other gain functions that could gauge the information that an attacker learns on a particular property of a secret. We note that $g$ is $\beta$-positive with $\beta = 1$. 

In comparison to Example \ref{ex:minNonDiff}, a distribution $\pP$ in $\Omega(\llrr{-1}{1})$ will now be characterised by two variables $\pP(0)$ and $\pP(1)$ since then $\pP(-1) = 1 - \pP(0) - \pP(1)$ will be fixed. The first variable $\pP(0)$ will take its values in $[0, 1]$ while the second one $\pP(1)$ will take its values in $[0, 1 - \pP(0)]$. 
We discretise the interval $[0, 1]$ into $101$ values so that $\pP(0)$ was assigned these values consecutively. For each of these values of $\pP(0)$, the interval $[0, 1-\pP(0)]$ is furthermore discretised into $101$ values that $\pP(1)$ took consecutively. 
For each pair $(\pP(0), \pP(1))$, we compute the value of $\obj_\infty^g(\pP)$ for the corresponding $\pP$, and we plot the resulting graph in Figure \ref{fig:minOp3D}.  
\iftoggle{bigFig}{%
\begin{figure}
\centering
\begin{tikzpicture}[scale=.9]
\begin{axis}[
  view={30}{25},
  width=2*175pt,
  tick label style={font=\scriptsize},
  axis lines=center,
  name=myplot,
	legend style={
	  at={(1,.7)},
	  anchor=south east,
	  draw=black,
	  fill=white,
	  legend cell align=left
	  },
  enlargelimits=.1,
  ymin=0, ymax=1, xmin=0, xmax=1, zmin=0, zmax=0.8,
  xlabel=$\pP(0)$, ylabel=$\pP(1)$, zlabel=$entropy$,
  every axis y label/.append style={at=(ticklabel* cs:0)}]
\addplot3[surf, opacity=0.5] file {op3D_30.dat};
\addlegendentry{$\obj_{\infty,g}(\pP)$}
\end{axis}
\end{tikzpicture}
\caption{Influence of $\pP$ in $\Omega(\llrr{-1}{+1})$ in the optimisation problem $\OP_{\infty,g}(1)$ in the computation of $f(x, y, z) = x(3y - 5z) + 2z$ with linear distributions over $\llrr{1}{30}$ for $\pT$ and $\pS$ and $\pA=\{5\colon 1\}$. }
\label{fig:minOp3D}
\end{figure}
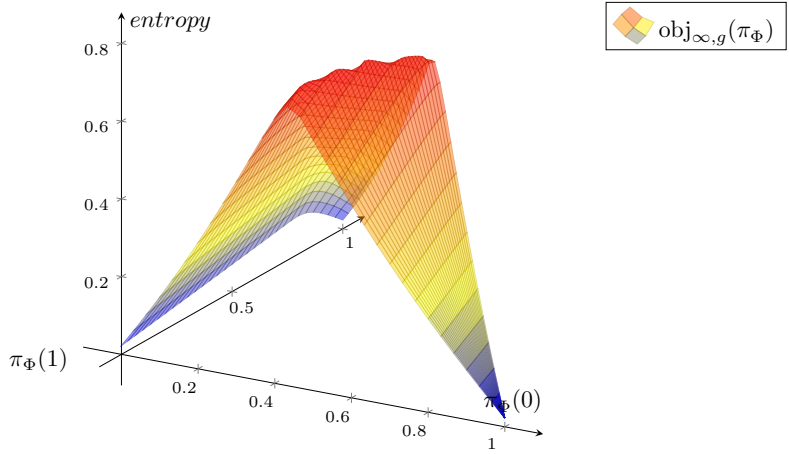
}{%
}

Let us now solve the optimisation problem $\OP_{\infty,g}$ with accuracy  $\epsilon = 10^{-2}$ through Method~ \ref{met:infinity}. Here, $|\cW|$ equals $2$. 
Let us then take $\alpha = \frac{2}{\epsilon} \cdot \log(|\cW|) = 200$. 
We then invoke Method~\ref{met:alpha} to solve $\OP_{\alpha, g}(\Delta)$ with accuracy $\epsilon' = \frac{\alpha}{\alpha - 1} \cdot \frac{\epsilon}{2} = 5.0 \cdot 10^{-3}$. 
Moreover, a combinatorial calculation gives us $|\DOp| = 5656$. 
We thus let:
\[
\beta = (1 - \frac{1}{\alpha}) \frac{\epsilon' \mu \cdot \ln 2}{|\DOp| \cdot |\cW|} = 3.0 \cdot 10^{-7}
\]

\noindent and we let $\bm{\delta}$ be the vector in $\Rpe^2$ whose components all equal $\beta$. 
Finally, we invoke Method \ref{met:delta} to solve 
$\OP_{\alpha,g}^{{\bm{\delta}}}(\Delta)$. 
This produces a nearly optimal solution $\pP_o = \{-1\colon 0.30, 0\colon 0.49, 1\colon 0.21\}$ for which $\obj_{\infty,g}(\pP_o)$ equals $0.77$. 
This ensures that $\omega_{\infty,g}$ is in $[0.77, 0.78]$ while a uniform distribution $\pP_u$ over $\{-1, 0, 1\}$ would have only $\obj_{\infty,g}(\pP_u) = 0.56$. 
\end{example}

\section{Discussion and Future Works}
\label{section:discussion}
In this work, we proposed an approach for quantifying the information that attackers can retrieve about private inputs from public outputs in black box computations of a public function. We also developed concepts and methods for mitigating against such information leakage, by distorting the public function with virtual, private inputs: we introduced some methods for maximising the posterior entropy of the targeted inputs, and developed non-linear optimisation techniques that can compute virtual inputs that optimally trade off the privacy protection stemming from virtual inputs and the accuracy of the distorted output in comparison with the un-distorted output. 

Our approach is generic in that, depending on the nature of the inputs and on the use context of the secure computation, the participants can agree on a particular type of entropy to maximise before entering the optimisation protocol. 
Participants may also want their inputs to be protected with respect to different kinds of entropy, and this could lead us to study multi-objective optimisation and Pareto optimality --- a topic for future work. 
The quantities and distortions that our approach can compute may also inform the risk management of using SCM for the same function repeatedly, with potentially different but related inputs --- such as the logging of daily health data. 

In a practical secure computation, once an optimal virtual distribution $\pP$ has been computed by our methods for a given type of entropy, the participants of the SMC would have to securely produce a virtual input drawn from distribution $\pP$. For example, the parties may enter an SMC protocol in order to produce a value $\phi$ that is secret to all the participants, and that follows distribution $\pP$. To that end, parties may generate locally shares of a virtual input such that the value obtained by the combination of these shares follows the specified distribution. Alternatively, it may also be practical to let a central authority compute the virtual inputs~---~and these virtual inputs could then be fed into SMC protocols in addition to the $x_i$ as seen on the right of Figure~\ref{fig:bb_model}. For example, if parties are health insurance providers, then the computation of virtual inputs by a central authority does not require any proof of compliance with health and data regulations, since the insurance providers would not share sensitive health data with that central authority. Designing such secure protocols is subject to future works. 

Our work considered the prior beliefs on the inputs to be public, constant, and part of the common knowledge. 
In SMC, this would enable participants to come to a consensus in order to agree on a common optimal virtual distribution $\pP$ and to securely compute the output of $f'$. 
In comparison to the setting of SMC which assumes that participants have agreed on a actively or passively secure protocol to use, we assume that participants in our case will agree on an approximate function $f'$ and a virtual distribution $\pP$ that protects the targets' privacy. 
In the case of outsourced computations, those public distributions could simply be used by a trusted third party in order to produce a virtual input drawn from $\pP$ and randomise the computation of $f$. 

On the other hand, it would be of interest to relax these assumptions.
In particular, computing an optimal virtual distribution $\pP$ requires having a prior belief $\pA$ on the attackers' input. Distribution $\pP$ would then maximise the targeted inputs' privacy given the prior belief $\pA$. But as the computation of $\pP$ can be performed offline by any of the parties, this could enable an attacker to substitute his input accordingly. This could thus update the belief $\pA$ and we would require another computation of $\pP$. 
The setting where two attackers would try to learn information about each other's input could also lead to interesting game-theoretic situations to be studied in future work. 

We also assumed that the partition of the participants into attackers, targets and spectators was given, but it would be of interest to develop techniques that can protect the participants' inputs when the set of potential attackers is not known. 
Moreover, we would like to further generalise our approach to the secure computation of vector-valued functions, i.e. of functions that compute several outputs, and where each of the outputs can be open to different sets of parties. 
Finally, scaling our approach to large input spaces is also one of our future research objectives. 

\section{Conclusion}
\label{section:conclusion}
Although efficient SMC protocols have been designed, information flow of outputs is inevitable, and has recently been rigorously formalised and quantified \cite{ah2017secure}. 
In this work, we first proposed a generalised notion of entropy that makes our approach compatible with different widely used measures of information flow. 
We then introduced the concepts of function substitution and virtual input that aim at randomising the output of SMC computations in order to impede the influence of deceitful attackers wishing to use input substitution to gain maximal information about private inputs from opened outputs. 
We have established some theoretical bounds for the privacy gain that approximations and close approximations provide. 
We then focused on additive approximations and formalised an optimisation problem that aims at maximising participants' privacy while controlling the distortion introduced on the output. 
We proposed different methods for solving such optimisation problems in practice and we experimentally showed that additive approximations give rise to significant privacy gains under specified distortion bounds.

\addcontentsline{toc}{section}{References}
\bibliographystyle{unsrt}
\bibliography{./references}
\end{document}